\documentclass[journal]{IEEEtran}

\usepackage[utf8]{inputenc}
\usepackage[T1]{fontenc}
\usepackage{cite}
\usepackage[caption=false,font=footnotesize]{subfig}
\usepackage{amsmath}
\usepackage{amssymb}
\usepackage{amsfonts}
\usepackage{placeins} 
\usepackage{url}
\usepackage{algorithm}
\usepackage[noend]{algpseudocode}
\usepackage{array}
\usepackage{booktabs}
\usepackage{breqn}
\usepackage{multirow}
\usepackage[table]{xcolor}
\usepackage{dblfloatfix}
\usepackage{graphicx}
\usepackage{float}
\usepackage{mathtools}
\usepackage{amsthm}
\newtheorem{proposition}{Proposition}
\newtheorem{remark}{Remark}
\allowdisplaybreaks

\begin{document}
\bstctlcite{IEEEexample:BSTcontrol}
\title{Dual-Orthogonality Waveforms \\ for Integrated Communication and Imaging \\ in Dynamic Multipath Channels}

\author{
Edoardo Talignani$^*$,~\IEEEmembership{Student Member,~IEEE}, Francesco Linsalata$^*$,~\IEEEmembership{Member,~IEEE}, 
\\ Musa Furkan Keskin,~\IEEEmembership{Senior Member,~IEEE}, 
Davide Scazzoli,~\IEEEmembership{Member,~IEEE}, Alireza Pourafzal,~\IEEEmembership{Member,~IEEE}, \\ Mohammad Mahdi Mojahedian,~\IEEEmembership{Member,~IEEE}, and Henk Wymeersch,~\IEEEmembership{Fellow,~IEEE}%
\thanks{E. Talignani, F. Linsalata, D. Scazzoli are with the Dipartimento di Elettronica, Informazione e Bioingegneria, Politecnico di Milano, Milano, Italy (e-mail: \{name.surname\}@polimi.it).}
\thanks{M. F. Keskin, A. Pourafzal, M. M. Mojahedian, and H. Wymeersch are with the Department of Electrical Engineering, Chalmers University of Technology, Gothenburg, Sweden (e-mail: {furkan@chalmers.se}, {alireza.pourafzal@chalmers.se}, {m.mojahedian@gmail.com}, {henkw@chalmers.se})}
\thanks{$^*$E. Talignani and F. Linsalata equally contributed to this research as co–first authors. The authors would like to thank Marco Manzoni and Stefano Tebaldini for their valuable discussions and insightful comments on the radar-related aspects of this work.}}

\maketitle

\begin{abstract}
Dual-Orthogonality waveforms are multi-antenna signaling schemes that enforce mutual orthogonality across transmit channels and over a prescribed set of delay shifts. By relaxing strict time orthogonality to the physically admissible region induced by propagation, these waveforms preserve full-band operation per transmit antenna while enabling the injection of communication data within the relaxed orthogonality structure and maintaining separability among spatial streams.
This property makes Dual-Orthogonality particularly attractive for Integrated Sensing and Communications (ISAC), where reliable data transmission, high-resolution sensing, and imaging must coexist under realistic, time-varying propagation conditions.
In dynamic communication multipath environments, combined delay--Doppler dispersion across multiple propagation paths perturbs the transmit subspaces and partially breaks the relaxed orthogonality conditions.
This paper analyzes the impact of multipath channel dynamics on Dual-Orthogonality waveforms and develops a multipath-structured decoding framework that exploits the waveform-induced subspace structure to enable low-complexity linear equalization.
{\color{black}
The receiver further incorporates structured multipath parameter estimation and effective-subspace reconstruction to support interference-aware decoding under realistic propagation conditions.
}
Numerical results show communication performance comparable to Orthogonal Frequency Division Multiplexing (OFDM)-based ISAC {\color{black}and MIMO-Orthogonal Time Frequency Space (OTFS) baselines} in time-varying multipath, while improving sensing/imaging via full-band per-transmit operation. The method achieves approximately 30~cm range resolution, $>15$~dB suppression of multipath imaging artifacts with coherent Synthetic Aperture Radar (SAR) processing, and a more favorable sensing--communication trade-off.
{\color{black}
Over-the-air experiments at 60~GHz validate complementary aspects of the framework: controlled measurements confirm multi-stream communication, the designed zero-correlation region, and accurate radar ranging, while a second campaign in a highly reflective indoor environment directly demonstrates multipath-aware stream equalization under strong unsuppressed reflections.
}
\end{abstract}


\section{Introduction}
Integrated sensing and communications (ISAC) extends the long-standing evolution of cellular networks beyond pure data transport. Early cellular generations supported secondary functions primarily in the form of user localization, typically with meter-level accuracy~\cite{delperal2018}. With the introduction of wider bandwidths and dedicated reference signals, recent standards have pushed positioning accuracy to the decimeter level~\cite{italiano2024tutorial}. ISAC generalizes this paradigm by enabling environmental sensing and imaging in addition to communication, using the same transmitted waveforms and shared physical resources~\cite{wei2023,saad2020}. In this context, sensing encompasses functionalities like target detection, ranging, tracking, and high-resolution imaging. This joint operation allows reliable multi-antenna data transmission while simultaneously supporting high-resolution sensing and imaging, which are expected to be native capabilities of future sixth-generation (6G) wireless systems.
\begin{figure}[!t]
    \centering
    \includegraphics[width=0.8\linewidth]{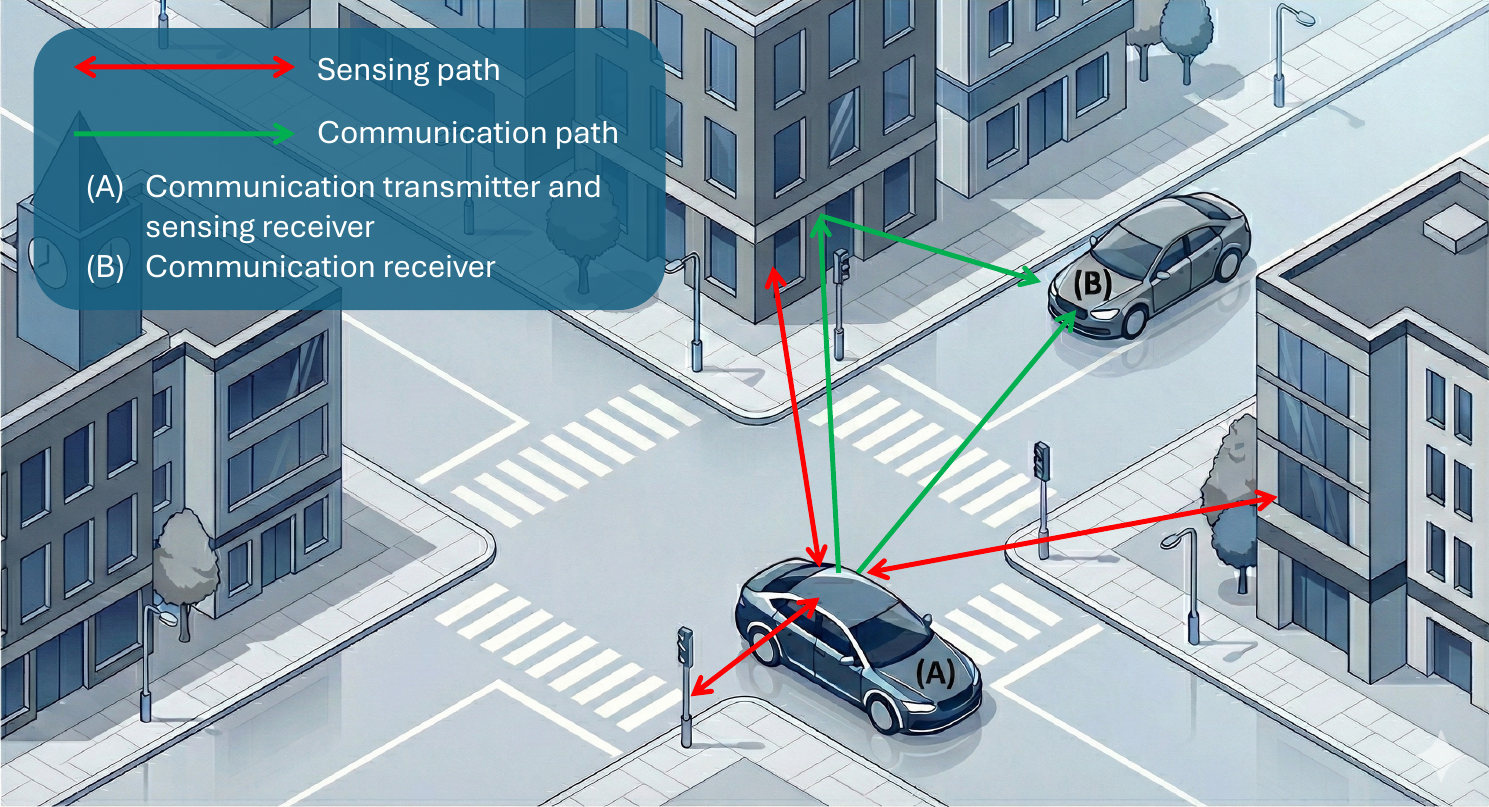}
    \caption{Representative ISAC scenario. Node~A simultaneously serves a remote communication receiver~B and senses the surrounding environment through reflected echoes. Multipath-induced interference degrades communication at~B, while providing exploitable information for sensing at~A.}
    \label{fig:scene}
\end{figure}
Fig.~\ref{fig:scene} highlights the fundamentally asymmetric role of multipath propagation in ISAC systems. While the sensing function inherently exploits reflections to reconstruct environmental features, the communication receiver experiences multipath as a source of interference. This dual role of multipath represents a central challenge in the design of practical ISAC waveforms and receivers, particularly under realistic delay and Doppler dynamics.

From a waveform design perspective, imaging extended and distributed targets ideally requires signals that remain orthogonal under arbitrary relative delay and Doppler shifts. In practice, however, realistic field-of-view and range constraints limit the set of physically observable delay values. This observation motivates a relaxation of strict time orthogonality to only those shifts that can be induced by the scene, allowing the remaining degrees of freedom to be exploited for data transmission. This principle underlies the concept of \emph{Dual-Orthogonality} waveforms, which are designed to be orthogonal across transmit antennas and over a prescribed set of admissible delay shifts \cite{ICaSSPManzoni}.

Multipath propagation plays a fundamentally asymmetric role in ISAC systems: it constitutes a primary source of information for sensing and imaging, while simultaneously degrading communication performance through inter-stream and inter-symbol interference.
Dual-Orthogonality waveforms were originally introduced in~\cite{ICaSSPManzoni},
demonstrating that relaxed orthogonality enables the joint support of
communication and high-resolution sensing. However, that work focused on
COSMIC waveform construction under idealized single-path communication
assumptions and did not explicitly address the impact of multipath channel
dynamics.
{\color{black}
That work did not characterize how multiple delay--Doppler paths perturb
the waveform-induced transmit subspaces, nor did it provide the channel
estimation, interference equalization, or coherent multipath-imaging
processing required for operation in dynamic propagation environments.
}
This paper closes this gap by explicitly analyzing Dual-Orthogonality
waveforms for communication in realistic multipath channels and by jointly
addressing waveform structure, receiver processing, channel estimation,
and environmental reconstruction.
{\color{black} These developments enable reliable
communication and high-resolution imaging under propagation conditions
that were outside the scope of the original COSMIC formulation.
}
In contrast to conventional ISAC receivers based on matched filtering, we show that, in the case of Dual-Orthogonality waveforms, multipath fundamentally alters the effective transmit subspaces observed at the receiver, rendering such approaches suboptimal. To address this issue, we develop a multipath-aware receiver framework based on structured channel estimation and low-complexity linear equalization, enabling reliable communication while preserving full-band sensing and imaging capabilities. {\color{black} Throughout the paper, Orthogonal Frequency Division Multiplexing (OFDM) and Orthogonal Time Frequency Space (OTFS) are adopted as reference baselines to benchmark the proposed approach under identical channel and bandwidth conditions.}

The main contributions of this work are as follows:
\begin{itemize}

{\color{black}
\item Building upon the original COSMIC waveform-construction principle
\cite{ICaSSPManzoni}, we derive a complete Dual-Orthogonality signal and
receiver model for dynamic multipath ISAC, explicitly characterizing the
subspace coupling introduced by multiple delay--Doppler propagation paths.
}

\item We analytically characterize the impact of multipath and Doppler
dynamics on the induced transmit subspaces and develop a structured
multipath-aware linear decoding strategy, supported by Cramér--Rao Bound
analysis to quantify estimation limits.
{\color{black}
The receiver formulation further includes practical estimation of
the number of dominant resolvable paths and explicitly characterizes the
linear combiner as a constrained Zero-Forcing subspace receiver, together
with its feasibility conditions and limitations.
}

\item Numerical simulations show communication performance comparable to
{\color{black}representative OFDM-based ISAC and MIMO-OTFS baselines}
while preserving full-band operation per transmit channel, enabling
decimeter-level range resolution (30~cm) and over $15\,\mathrm{dB}$
multipath artifact suppression.
{\color{black}
The comparison jointly considers communication performance, imaging
quality, and effective goodput under matched dynamic multipath
conditions.
}
Joint ISAC results confirm operation in a favorable
sensing--communication trade-off region with competitive effective goodput.

\item We experimentally validate the proposed waveforms at
$60\,\mathrm{GHz}$, confirming reliable multi-stream communication
(Error Vector Magnitude (EVM) $\approx -13\,\mathrm{dB}$), a measurable
zero-correlation region, and centimeter-level radar ranging accuracy.
{\color{black}
Moreover, an experimental campaign conducted in a reflective indoor environment directly evaluates the multipath-aware communication receiver under strong unsuppressed reflections, demonstrating substantial
EVM recovery through the proposed subspace equalization without relying on receive-side spatial separation.
}

\end{itemize}

The remainder of this paper is organized as follows. Section~II reviews related work on ISAC waveform design. Section~III introduces the system and signal model. Section~IV presents the proposed multipath-aware communication receiver and channel estimation framework. Section~V addresses imaging. Section~VI reports numerical simulation results, while Section~VII provides experimental validation. Section~VIII concludes the paper.

\subsubsection*{Paper Notation}
Scalars are denoted by italic letters, vectors by bold lowercase letters, and
matrices by bold uppercase letters. Continuous-time signals are denoted by
$(\cdot)$, while discrete-time samples are denoted by $[\cdot]$.
Superscripts $n$ index transmit antennas or data streams, whereas subscripts
index propagation paths or auxiliary quantities.
The operators $(\cdot)^{\top}$, $(\cdot)^{\mathsf{H}}$, and $\mathrm{vec}(\cdot)$ denote
transpose, Hermitian transpose, and vectorization, respectively, and
$\otimes$ denotes the Kronecker product. The operator $\mathrm{diag}(x_k)_{k=0}^{K-1}$ denotes a $K \times K$ diagonal matrix whose $k$-th main diagonal element is $x_k$.
$\mathbb{C}$ denotes the complex field.  $\mathbf{I}_K$ denotes the $K\times K$ identity
matrix, while $\mathbf{0}$ the all-zero vector.

\section{Related Works}

\begin{table}[!t]
\centering
\caption{\color{black}Summary of Representative ISAC Waveform Categories}
\label{tab:ISAC_waveforM_Summary}
\resizebox{0.5\textwidth}{!}{%
\begin{tabular}{l p{4cm} p{4cm} l}
\toprule
\textbf{Waveform} & \textbf{Pros} & \textbf{Cons} & \textbf{Ref.} \\ 
\midrule

\multicolumn{4}{l}{\textit{Communication-centric}} \\ 
CP-OFDM & Standardized; spectrally efficient & High PAPR; Doppler-sensitive & \cite{3gpp38211,liu2022fundlimits,yuan2023multipleCP} \\ 

CP-SC & Low PAPR; PA-efficient & Equalization complexity & \cite{zhou2022,tanahashi2010scpower,zeng2020cpsc} \\ 

CAZAC & Ideal autocorrelation & Doppler sensitive; limited modulation & \cite{zhang2024cazac} \\ 

OTFS & Robust in delay--Doppler & High complexity & \cite{shtaiwi2024otfssurvey,gaudio2020otfsisac,raviteja2019otfsradar,keskin2024mimootfs} \\ 

AFDM & Full diversity; multipath resolution & Chirp optimization needed & \cite{bemani2024integrated,liu2024afdmPIM} \\ 

Filtered (UFMC/FBMC) & Low OOB leakage & Intrinsic interference; complex pilots & \cite{zhou2022} \\ 

\midrule
\multicolumn{4}{l}{\textit{Sensing-centric}} \\ 

LFM/FMCW & High range resolution; Doppler robust & Low data rate; sweep trade-off & \cite{liu2022jsac_1,zhou2022} \\ 

Index Modulation & Radar-compatible data embedding & Detection complexity; index errors & \cite{Elbir2024IMISAC_SPM,Xu2020GSM,DoganTusha2021MDIM} \\ 

Spatial/PPM & Constant envelope; hardware efficient & Low throughput; delay ambiguity & \cite{hassanien2016overview,liu2025ppm,zhou2022} \\ 

\midrule
\multicolumn{4}{l}{\textit{Joint-design}} \\ 

Cooperative ISAC & Rate--accuracy Pareto optimization & Coordination overhead & \cite{GaoLianYu_JSAC_2023_CoISAC_RSMA_Pareto} \\ 

FD Joint Beamforming & Joint SI mitigation and beamforming & Strong coupling; CSI dependent & \cite{HeEtAl_JSAC_2023_FD_ISAC_Joint_Beamforming} \\ 

Symbol-Level Precoding & Symbol-level rate--sensing trade-off & High per-symbol complexity & \cite{LiuEtAl_JSTSP_2021_SLP_DFRC} \\ 

COSMIC & Joint communication–imaging & CSI dependent & \cite{ICaSSPManzoni} \\ 

\bottomrule
\end{tabular}%
}
\end{table}

Waveform designs for ISAC can be broadly categorized based on how communication and sensing functionalities are combined. 
Multipath propagation plays a central role: it is often treated as an impairment to be mitigated or suppressed, even though it also constitutes a fundamental source of information for sensing and imaging.

\textbf{Communication-oriented Designs} exploit the fact that known communication signals can be correlated at the receiver for sensing using pilots, preambles, or data-aided reconstructions~\cite{liu2022jsac_1,zhou2022}.
As a result, design priorities include favorable correlation properties, robustness to Doppler and carrier-frequency offsets, and low peak-to-average power ratio (PAPR). OFDM, which underpins the 5G New Radio (NR) physical layer~\cite{3gpp38211}, is the most prominent example due to its Fast Fourier Transform (FFT)-based transceivers, high spectral efficiency, and natural multiple-input multiple-output (MIMO) scalability~\cite{liu2022fundlimits}. Nevertheless, Cyclic Prefix (CP)--OFDM suffers from high PAPR, sinc-shaped sidelobes, Doppler sensitivity, and cyclic-prefix overhead~\cite{yuan2023multipleCP}. To better handle doubly selective channels, delay--Doppler-domain multicarrier schemes such as OTFS and Affine Frequency Division Multiplexing (AFDM) have been proposed. These schemes exploit sparse delay--Doppler representations to resolve multipath components and mitigate Doppler-induced interference, enabling accurate joint delay and Doppler estimation at the cost of increased transceiver complexity and sensitivity to pulse design~\cite{shtaiwi2024otfssurvey,bemani2024integrated}. Extensions such as AFDM-based index modulation further embed information onto additional index dimensions while preserving sensing performance~\cite{liu2024afdmPIM}. Other communication-centric candidates include single-carrier waveforms, constant-amplitude zero-autocorrelation (CAZAC) sequences, and filtered multicarrier schemes (e.g., UFMC, FBMC/OQAM, GFDM), which offer improved power efficiency or spectral containment at the expense of increased receiver complexity or reduced Doppler robustness~\cite{zhou2022}.

\textbf{Sensing-oriented Designs} primarily prioritize the radar waveform, while
communication information is embedded so as not to disrupt sensing
performance.
Typical implementations rely on linear frequency-modulated (LFM) or frequency-modulated continuous-wave (FMCW) chirps, as well as phase- or code-modulated radar signals, preserving constant-envelope transmission, ambiguity-function control, and Doppler tolerance~\cite{liu2022jsac_1,zhou2022}. Information embedding is commonly achieved via index modulation across antennas, subcarriers, pulses, or codes, allowing data transmission while maintaining radar compatibility~\cite{Elbir2024IMISAC_SPM,Xu2020GSM,DoganTusha2021MDIM}. Alternative approaches encode information into spatial sidelobes of the transmit beampattern or into controlled timing shifts, as in pulse-position modulation (PPM), enabling concurrent radar operation and low-rate communication~\cite{hassanien2016overview,liu2025ppm}. While hardware efficient, these schemes face intrinsic limits on achievable data rates and may introduce parameter coupling between communication symbols and sensing observations, potentially biasing range--Doppler estimates~\cite{zhou2022}.

\textbf{Joint Designs} depart from simple waveform reuse and instead adopt
system-level optimization, designing signals or beamformers that explicitly
balance communication and sensing objectives.
Representative examples include cooperative multi-base-station formulations that map the Pareto frontier between communication rate and localization accuracy~\cite{GaoLianYu_JSAC_2023_CoISAC_RSMA_Pareto}, full-duplex ISAC architectures with joint beamforming and power optimization under radar Signal-to-Interference-plus-Noise Ratio (SINR)  constraints~\cite{HeEtAl_JSAC_2023_FD_ISAC_Joint_Beamforming}, and symbol-level precoding for dual-functional radar--communication systems, where constant-modulus waveforms are optimized on a per-symbol basis to meet both sensing and communication requirements~\cite{LiuEtAl_JSTSP_2021_SLP_DFRC}. Although powerful, these methods typically rely on centralized optimization, accurate channel state information (CSI), and high computational complexity, which can limit scalability and real-time applicability.

\textbf{Discussions:} Despite their differences, existing ISAC waveform designs largely rely on idealized propagation assumptions or receiver structures that do not explicitly account for multipath-induced subspace coupling, particularly in the context of joint communication and imaging. This limitation motivates waveform and receiver designs that remain effective under realistic channel dynamics while preserving high-resolution sensing and imaging capabilities. Table~\ref{tab:ISAC_waveforM_Summary} provides a comparative summary of representative ISAC waveform categories, highlighting their main advantages and limitations.

OFDM-based schemes are particularly relevant for ISAC due to their direct compatibility with 5G New Radio (NR) and existing cellular infrastructures. 
In practical
MIMO--OFDM sensing and ISAC systems, multi-antenna transmission is commonly realized
via frequency-division multiple access (FDMA), whereby each transmitter occupies a
disjoint subset of subcarriers. While this preserves inter-antenna orthogonality and
enables joint delay--Doppler--angle estimation using standard OFDM processing~\cite{keskin2021jstsp},
it effectively reduces the per-transmitter bandwidth. When echoes are processed
independently, for instance in time-domain back-projection imaging, this leads to a
range-resolution loss proportional to the number of transmit antennas unless coherent
combination is performed, a known limitation of FDMA MIMO radar~\cite{cohen2020fdma}.

Interleaved OFDM (I--OFDM) mitigates this loss by assigning each transmitter a uniformly
interleaved set of subcarriers spanning the full bandwidth, thereby preserving range
resolution while maintaining frequency-domain orthogonality~\cite{Sturm2013SpectralInterleave}.
However, the resulting periodic spectral structure produces Dirichlet-like
autocorrelation functions with grating sidelobes at predictable delay offsets, which
can degrade imaging performance~\cite{Lin2015IOFDM}. Non-equidistant subcarrier
interleaving (NeqSI) addresses this issue by employing a jointly designed non-uniform
spectral lattice that breaks periodicity and suppresses coherent sidelobe buildup,
while preserving the mainlobe determined by the total bandwidth~\cite{Hakobyan2016NeqSI,Hakobyan2020NeqSI_Optimized,Wu2022OFDM_Tutorial}.


Despite these advantages, both I--OFDM and NeqSI introduce practical challenges.
Irregular spectral allocations increase transceiver complexity and impose stringent
inter-antenna synchronization requirements under timing and carrier-frequency
offsets~\cite{Sturm2013SpectralInterleave,Lin2015IOFDM}. In addition, NeqSI is sensitive
to frequency selectivity, as non-uniform tone and pilot placement complicates channel
estimation and may degrade decoding performance in rapidly time-varying channels~\cite{Hakobyan2020NeqSI_Optimized}.
Both schemes also remain susceptible to Doppler-induced intercarrier interference if
compensation is imperfect~\cite{Lin2015IOFDM}. Existing works primarily treat NeqSI as
a radar-oriented extension of communication-style multicarrier signaling rather than
as part of a unified ISAC design~\cite{Hakobyan2016NeqSI,Hakobyan2020NeqSI_Optimized}.
Nevertheless, because NeqSI yields superior radar performance while preserving the robust communication capabilities of standard OFDM, despite its higher implementation complexity, it is adopted here{\color{black}, alongside OTFS as the state-of-the-art baseline for delay--Doppler resilience,} and used for comparison with the proposed Dual-Orthogonality waveform scheme under identical conditions. 

\section{System and Signal Model}
\label{sec:sys_signal_model}

We consider an ISAC node equipped with $N_T$ transmit and $M_S$ receive antennas. The moving ISAC node transmits Dual-Orthogonality waveforms that simultaneously support data transmission toward a remote communications receiver and probing of the surrounding environment
through reflected echoes. Although all observations originate from the same transmitted signal, the resulting received signals may differ due to distinct
receiver locations, scattering geometries, and propagation conditions.

\subsection{Transmit Signal Model}
\label{subsec:tx_model}

The ISAC transmitter operates over a bandwidth $B$ during a signaling interval of
duration $T_p$. Let $T_s=1/B$ denote the sampling period and define
$K \triangleq T_p / T_s$. The discrete-time transmit signal across the $N_T$
transmit antennas is given by
$\mathbf{s}[k]
=
\big[\, s_1[k],\; s_2[k],\; \ldots,\; s_{N_T}[k] \,\big]^{\top}$
with $k = 0,1,\ldots,K-1$, where $s_n[k] \in \mathbb{C}$ denotes the sample
transmitted by the $n$-th antenna. Collecting the samples over the signaling
interval yields the transmit block
$\mathbf{S}
=
\big[\, \mathbf{s}_1,\mathbf{s}_2,\ldots,\mathbf{s}_{N_T} \,\big]
\in \mathbb{C}^{K\times N_T}$,
where $\mathbf{s}_n \in \mathbb{C}^{K\times 1}$ denotes the waveform transmitted by
the $n$-th antenna.

The waveform design follows the principle of \emph{Dual-Orthogonality} \cite{ICaSSPManzoni},
which enforces signal separability along two complementary dimensions. First,
signals transmitted from different antennas are assigned to mutually orthogonal
subspaces of the $K$-sample signal space, ensuring inter-stream orthogonality in
the absence of propagation effects. Second, this orthogonality is enforced only
over a finite and physically admissible delay window determined by the maximum
extent of the sensed scene, while the remaining degrees of freedom are exploited
for data transmission.

To enforce spatial orthogonality, each waveform is assigned to a $K_s$-dimensional
orthogonal subspace of the $K$-sample signal space, with $K_s<K$ here set as $K_s = $ $\lfloor K / N_T \rfloor$. Let
$\mathbf{C} \in \mathbb{C}^{K\times K}$ denote a unitary basis
($\mathbf{C}^{\mathsf{H}}\mathbf{C} = \mathbf{I}_K$), and let
$\mathbf{C}_n \in \mathbb{C}^{K\times K_s}$ collect a distinct subset of columns of $\mathbf{C}$
assigned to the $n$-th antenna. The transmitted waveform associated with the
$n$-th antenna is expressed as
\begin{equation}\label{general_structure}
\mathbf{s}_n = \mathbf{C}_n\,\boldsymbol{\alpha}_n,
\qquad n = 1,\ldots,N_T,
\end{equation}
where $\boldsymbol{\alpha}_n \in \mathbb{C}^{K_s\times 1}$ contains the symbol space associated with the $n$-th stream. The assigned
subspaces are mutually orthogonal, i.e.,
$\mathbf{C}_i^{\mathsf{H}}\mathbf{C}_j=\delta_{ij}\mathbf{I}_{K_s}$, ensuring ideal
inter-antenna orthogonality under single-path propagation. 

In multipath environments, delayed replicas break subspace separation. To address this effect, orthogonality is enforced only over a finite delay interval corresponding to the maximum round-trip propagation time within the scene. Let $\widetilde{\mathbf{S}}_i\in\mathbb{C}^{K_z\times K}$ denote the reduced cross-correlation (Toeplitz) operator associated with the $i$-th transmit
subspace, which models correlation with all delayed replicas (up to $K_z$ samples) of the waveform associated with $\mathbf{C}_i$. 
{\color{black} 
To account for the zero-delay tap, $K_z = \left\lceil \tau_{s} / T_s \right\rceil + 1$ corresponds to the maximum physically admissible round-trip delay of the sensed scene, with $\tau_{s} = 2R_{s}/c$, where $R_{s}$ denotes the maximum sensing range and $c$ is the speed of light.
}Note that the orthogonality construction requires a bound of $K_z<K_s$.
Specifically, the rows
of $\widetilde{\mathbf{S}}_i$ span the first $K_z$ delayed versions of $\mathbf{s}^{\mathsf{H}}_i$~\cite{ICaSSPManzoni}.
Dual-Orthogonality then requires that
\begin{equation}
\widetilde{\mathbf{S}}_i\,\mathbf{s}_n=\mathbf{0},
\qquad \forall\, i<n,
\end{equation}
i.e., the waveform transmitted by the $n$-th antenna is orthogonal, over the admissible
delay window, to the subspace spanned by all previously assigned transmit waveforms. 
Collecting these constraints yields the linear operator
\begin{equation}
\mathbf{Q}_n
\triangleq
\begin{bmatrix}
\widetilde{\mathbf{S}}_1\mathbf{C}_n \\
\widetilde{\mathbf{S}}_2\mathbf{C}_n \\
\vdots \\
\widetilde{\mathbf{S}}_{n-1}\mathbf{C}_n
\end{bmatrix}\in \mathbb{C}^{(n-1)K_z\times K_s},
\end{equation}
with the generic block $\widetilde{\mathbf{S}}_l\mathbf{C}_n \in \mathbb{C}^{K_z\times K_s}$. As $\widetilde{\mathbf{S}}_l$ is built from $\mathbf{s}_l$, $\mathbf{Q}_n$ directly depends on the preceding symbols $\boldsymbol{\alpha}_1, \dots, \boldsymbol{\alpha}_{n-1}$, resulting in a data-dependent construction. Note that for the first stream ($n=1$), the constraint set is empty, leaving the assigned subspace fully available. {\color{black} 
To prevent an empty null space for higher-index antennas, the available degrees of freedom must strictly bound $K_z$. Guaranteeing a strictly positive null-space dimension $d_n = K_s - (n-1)(K_z - 1) > 0$ for all $N_T$ streams requires $K_z \le \lfloor (K_s - 1)/(N_T - 1) \rfloor + 1$. 

For instance, considering a safe sensing margin of $R_s = 20$~m (corresponding to a 40~m round-trip distance, yielding $\tau_s \approx 0.133~\mu\mathrm{s}$), the $B=40$~MHz configuration yields $K=1600$, $K_s=200$, and $K_z=7$, ensuring strictly positive null-space dimensions $d_n \in \{200, 194, 188, 182, 176, 170, 164, 158\}$ for all $N_T=8$ streams. Similarly, $B=200$~MHz yields $K=8000$, $K_s=1000$, and $K_z=28$, resulting in $d_n \in \{1000, 973, 946, 919, 892, 865, 838, 811\}$.
}

The projected symbol vector $\boldsymbol{\alpha}_n$ is thus constrained to lie in the null space of $\mathbf{Q}_n$, i.e., $\boldsymbol{\alpha}_n \in \mathrm{Null}\!\left(\mathbf{Q}_n\right)$.  Let the columns of $\mathbf{N}_n\in \mathbb{C}^{K_s\times (K_s-(n-1)(K_z-1))}$ form an orthonormal basis for
$\mathrm{Null}\!\left(\mathbf{Q}_n\right)$.
The transmitted coefficients are then
synthesized as
\begin{equation}
\boldsymbol{\alpha}_n = \mathbf{N}_n \mathbf{x}_n,
\end{equation}
where $\mathbf{x}_n \in \mathbb{C}^{K_s-(n-1)(K_z-1)}$ denotes the vector of information symbols for the $n$-th stream. The ``$(K_z-1)$'' term accounts for the zero-delay constraint: since $\mathbf{C}_i^{\mathsf{H}}\mathbf{C}_n = \mathbf{0}$ by construction, the corresponding row in $\mathbf{Q}_n$ identically vanishes and does not consume a degree of freedom.

\begin{remark}
An important feature of this construction is its inherent flexibility. The choice of the unitary basis $\mathbf{C}$ allows the waveform family to adapt to different architectural requirements without modifying the underlying signal model. For instance, selecting $\mathbf{C}$ as a Discrete Fourier Transform (DFT) matrix reduces the model to conventional OFDM transmission. More generally, when $\mathbf{C}$ is instantiated as a random
unitary operator, the resulting waveform structure provides implicit physical-layer obfuscation by avoiding publicly fixed signaling bases.
\end{remark}

\subsection{Propagation Model}
\label{subsec:comm_channel}

We model the propagation between the ISAC transmitter and a generic receiver as a
$P$-path narrowband MIMO channel. In this subsection, we focus on the observation
collected at a remote receiver equipped with $M_C$ antennas. Let $\mathbf{y}(t)\in\mathbb{C}^{M_C\times 1}$ denote the continuous-time baseband signal
observed at this receiver. The received signal is expressed as
\begin{equation}\label{comm_rec_equation}
\mathbf{y}(t)
=
\sum_{p=0}^{P-1}
\gamma_p\, e^{j 2\pi \nu_p t}\,
\mathbf{a}_{\mathrm{rx}}(\phi_p)\,
\mathbf{a}_{\mathrm{tx}}^{\top}(\theta_p)\,
\mathbf{s}(t-\tau_p)
+
\mathbf{w}(t),
\end{equation}
where $\gamma_p \in \mathbb{C}$, $\tau_p \in \mathbb{R}$, and
$\nu_p \in \mathbb{R}$ denote the complex gain, propagation delay, and Doppler
shift of the $p$-th path, respectively, while $\theta_p$ and $\phi_p$ denote its
angle of departure and angle of arrival, respectively. The vectors
$\mathbf{a}_{\mathrm{tx}}(\cdot)\in\mathbb{C}^{N_T\times 1}$ and
$\mathbf{a}_{\mathrm{rx}}(\cdot)\in\mathbb{C}^{M_C\times 1}$ represent the transmit
and receive array steering vectors, and $\mathbf{w}(t)$ is additive white Gaussian
noise. Sampling the received signal at $K$ instants over the signaling interval yields the
discrete-time observation
$\mathbf{Y} \in \mathbb{C}^{K\times M_C}$. 
Under the standard narrowband delay–Doppler channel model\footnote{{The narrowband assumption is valid when the maximum delay drift over the pulse duration is negligible compared to the sampling period, i.e., $2 v_{\max} T_p / c \ll T_s$. Under the simulated mobility ($v_{\max} = 30$ m/s, $T_p = 40\ \mu$s), the maximum intra-pulse delay drift is 0.008 ns. This strictly satisfies the condition even for the most restrictive simulated bandwidth of $B = 200$ MHz, where $T_s = 5$ ns.}}, where fractional delays are represented through frequency-domain phase shifts and Doppler induces a discrete-time phase rotation without time-scaling of the baseband waveform, sampling \eqref{comm_rec_equation} at $K$ time instants yields the following compact matrix representation\cite{keskin2024mimootfs}:
\begin{equation}\label{final_model}
\hspace{-0.2cm}\mathbf{Y}
=
\sum_{p=0}^{P-1}
\gamma_p
\mathbf{D}(\nu_p)
\mathbf{F}^{\mathsf{H}}
\mathbf{B}(\tau_p)
\mathbf{F}
\mathbf{S}
\mathbf{a}_{\mathrm{tx}}(\theta_p)
\mathbf{a}^{\top}_{\mathrm{rx}}(\phi_p)
+
\mathbf{W},
\end{equation}
where $\mathbf{F}\in\mathbb{C}^{K\times K}$ denotes the unitary DFT matrix,
$\mathbf{B}(\tau_p)$ and $\mathbf{D}(\nu_p)\in\mathbb{C}^{K\times K}$ are diagonal
matrices defined as $\mathbf{B}(\tau_p)= \mathrm{diag}\!\left(e^{-j2\pi k \tau_p / (K T_s)}\right)_{k=0}^{K-1}$ and $\mathbf{D}(\nu_p)= \mathrm{diag}\!\left(e^{j2\pi k \nu_p T_s}\right)_{k=0}^{K-1}$, with
$k = 0,\ldots,K-1$, modeling continuous-valued delay via frequency-domain phase shifts (in the discrete-time Fourier basis) and Doppler via time-domain phase rotations, and
$\mathbf{W}\in\mathbb{C}^{K\times M_C}$ denotes additive noise.

{\color{black}
While physical multipath propagation induces linear (Toeplitz) shifts, the model in \eqref{final_model} employs a circulant matrix approximation $\mathbf{F}^{\mathsf{H}}\mathbf{B}(\tau_p)\mathbf{F}$ commonly used for compact receiver modeling \cite{keskin2024mimootfs}. Because Dual-Orthogonality waveforms do not employ a Cyclic Prefix (CP), this circulant operator constitutes an approximation of the true physical channel. However, because the signaling interval is significantly longer than the maximum channel delay spread ($T_p \gg \tau_{\max}$), the difference between the linear and circular shifts is strictly confined to a small fraction of the block boundary. As formally proven in Appendix \ref{app:circulant_bound}, the relative error energy introduced by this approximation is rigorously bounded by $2\tau_{\max}/T_p \ll 1$, thereby safely justifying the circulant approximation for equalization purposes.}

\begin{remark}
An analogous model applies to the sensing receiver co-located with the transmitter. In the monostatic case, the propagation paths correspond to target reflections and multipath scattering. The same delay--Doppler parameterization therefore provides a unified description for both communication and sensing observations.
\end{remark}

\section{Communication Receiver Processing}

This section addresses receiver-side processing at the \emph{remote
communication receiver}. Based on the signal model introduced in Sec.~III,
we analyze how multipath propagation alters the ideal orthogonality properties
of Dual-Orthogonality waveforms and how explicit estimation of the underlying
multipath components enables reliable data decoding through multipath-aware
linear equalization. The communication receiver is equipped with a uniform
linear array (ULA) of $M_C$ antennas with half-wavelength inter-element spacing,
which enables spatial filtering and subspace-based processing.

\subsection{Orthogonality Loss and Effective Channel}
\label{subsec:loss_orthogonality}

Let $\mathbf{y}_m \in \mathbb{C}^{K\times1}$ denote the signal observed at the $m$-th antenna of the remote receiver, consisting of the $m$-th column of~(\ref{final_model}). 
The decoding of stream $n$ inherently depends on the structure induced by the previously assigned streams $1,\dots,n-1$. Indeed, as explained in Sec.~III-A, the Dual-Orthogonality subspaces are constructed sequentially, so that each new subspace is defined under orthogonality constraints with respect to the previously assigned ones.
Decoding of the $n$-th transmit stream naturally starts by projecting the received signal onto the corresponding transmit subspace,
\begin{equation} \label{eq:alpha_tilde}
\widetilde{\boldsymbol{\alpha}}_{n,m}
=
\mathbf{C}^{\mathsf{H}}_n \mathbf{y}_m .
\end{equation}
Substituting, from Sec.~III, the transmit structure \eqref{general_structure} into \eqref{final_model} and applying the projection in \eqref{eq:alpha_tilde} to the observation from the m-th receive antenna, we obtain
\begin{equation}
\hspace{-0.2cm}\widetilde{\boldsymbol{\alpha}}_{n,m}
\hspace{-0.15cm}=
\hspace{-0.15cm}\sum_{p=0}^{P-1}\hspace{-0.1cm}\sum_{\ell=1}^{N_T}
\hspace{-0.17cm}\gamma_p
[\mathbf{a}_{\mathrm{rx}}(\phi_p)]_m
[\mathbf{a}_{\mathrm{tx}}(\theta_p)]_\ell
\mathbf{C}^{\mathsf{H}}_n\mathbf{H}_p\mathbf{C}_\ell
\hspace{-0.02cm}\boldsymbol{\alpha}_\ell \hspace{-0.07cm}+
\hspace{-0.08cm}\mathbf{C}^{\mathsf{H}}_n\mathbf{w}_m,
\end{equation}
where $\mathbf{H}_p= \mathbf{D}(\nu_p)\,\mathbf{F}^{\mathsf{H}}\mathbf{B}(\tau_p)\mathbf{F}$ denotes the delay--Doppler operator associated with the
$p$-th propagation path, consistently with the sampled channel model in \eqref{final_model}.
In the absence of delay and Doppler dispersion, the orthogonality property
$\mathbf{C}^{\mathsf{H}}_n\mathbf{C}_\ell=\mathbf{0}$ for $\ell\neq n$ ensures perfect
stream separation. In general multipath channels, however, the operators
$\mathbf{H}_p$ rotate the transmit subspaces, generating inter-antenna
interference (IAI) and intra-stream distortion.
To better understand the impact, we separate this expression into three distinct components as

\begin{equation}
\widetilde{\boldsymbol{\alpha}}_{n,m}
=
\underbrace{\!\left(\mathbf{C}^{\mathsf{H}}_n\mathbf{G}_{n,m}\right)\!\boldsymbol{\alpha}_n}_{\text{Distorted Signal}}
+\underbrace{\!\sum_{\ell \neq n}\!\left(\mathbf{C}^{\mathsf{H}}_n\mathbf{G}_{\ell,m}\right)\!\boldsymbol{\alpha}_\ell}_{\text{Inter-Antenna Interference}}
+\hspace{-0.15cm}\underbrace{\mathbf{C}^{\mathsf{H}}_n\mathbf{w}_m}_{\text{Filtered Noise}}
\label{eq:alphan_decomposed}
\end{equation}

where $\mathbf{G}_{\ell,m}$ will be introduced shortly. 
This decomposition highlights how multipath-induced delay--Doppler operators rotate the transmit subspaces, generating IAI. 
The expression in \eqref{eq:alphan_decomposed} serves to explicitly identify these contributions and motivate the structured equalization strategy developed in the following subsection.
To capture this effect, we separate in \eqref{eq:alpha_tilde} the desired contribution 
corresponding to $\ell = n$ from the interference terms $\ell \neq n$. 
Focusing on the component associated with stream $n$, we define the 
\emph{effective subspace channel} of the $n$-th stream and $m$-th receiving antenna as
\begin{equation}\label{InterferenceModel}
\mathbf{G}_{n,m}
\triangleq
\sum_{p=0}^{P-1}
\gamma_p
[\mathbf{a}_{\mathrm{rx}}(\phi_p)]_m
[\mathbf{a}_{\mathrm{tx}}(\theta_p)]_n
\mathbf{H}_p\mathbf{C}_n.
\end{equation}
Note that $\mathbf{G}_{n,m} \in\;\mathbb{C}^{K\times K_s}$ does not represent the full MIMO channel matrix, but the channel induced on the $n$-th transmit subspace by the physical multipath operator, capturing the delay--Doppler distortion within that subspace.
Reliable decoding therefore hinges on accurate estimation of the multipath
parameters that define $\mathbf{G}_{n,m}$.
The effective channel $\mathbf{G}_{n,m}$ can be interpreted as the original transmit subspace $\mathbf{C}_n$ undergoing a sequence of unitary delay and Doppler transformations, weighted by the path gain and transmit-array response. As a result, multipath propagation maps each transmit subspace into a rotated and distorted subspace at the receiver, whose superposition across paths defines the effective channel seen by stream $n$.
Assuming the effective subspace channel $\mathbf{G}_{n,m}$ is known, we first describe the linear equalization strategy that enables stream separation under multipath-induced subspace coupling. The subsequent subsection details how the required multipath parameters are estimated in practice.

\subsection{Multipath-Aware Equalization}

\label{subsec:equalization}
As discussed in the previous subsection, projection onto the original transmit subspaces via matched filtering does not account for multipath-induced subspace distortion and therefore cannot fully suppress IAI. We therefore introduce a multipath-aware linear equalization strategy that explicitly incorporates the effective channel structure into the combiner design.
Using the effective channel model, decoding of the $n$-th stream at the $m$-antenna is performed via linear equalization. A linear combiner $\mathbf{A}_{n,m}\in\mathbb{C}^{K\times K_s}$ produces the subspace estimate
{\color{black}
\begin{equation}
\widehat{\boldsymbol{\alpha}}_{n,m}
=
\mathbf{A}^{\mathsf{H}}_{n,m}\mathbf{y}_m,
\end{equation}
}
{\color{black} where $\mathbf{y}_m \in \mathbb{C}^{K \times 1}$ denotes the received signal vector corresponding to the $m$-th column of the sampled observation matrix $\mathbf{Y} \in \mathbb{C}^{K \times M_C}$ in \eqref{final_model}.
}

\begin{remark}
In the special case where a single path $\hat{p}$ dominates, 
the effective channel $\mathbf{G}_{\ell,m} $ reduces to the contribution 
of that path only, as $\mathbf{G}^{\hat{p}}_{\ell,m}$. 
Inter-stream leakage can therefore be suppressed by enforcing
$\mathbf{A}^{\mathsf{H}}_{n,m}\mathbf{G}^{\hat{p}} _{\ell,m}=\mathbf{0},
\quad \ell \neq n$, 
where $\mathbf{G}^{\hat{p}}_{\ell,m}
=
\gamma_{\hat{p}} 
[\mathbf{a}_{\mathrm{rx}}(\phi_{\hat{p}})]_m
[\mathbf{a}_{\mathrm{tx}}(\theta_{\hat{p}})]_\ell
\mathbf{H}_{\hat{p}}\mathbf{C}_{\ell}$,
for $\ell \neq n$.
While effective in near single-path conditions, this design becomes suboptimal in
multipath-rich environments. 
Since $\mathbf{H}_{\hat{p}}$ is rank preserving (see Appendix~\ref{app:rank_proof}), each matrix $\mathbf{G}^{\hat{p}} _{\ell,m} \in \mathbb{C}^{K \times K_s}$ 
has rank $K_s$. 
The interference matrix formed by concatenating 
$\{\mathbf{G}^{\hat{p}}_{\ell,m}\}_{\ell \neq n}$ 
has full column rank $(N_T-1)K_s$ and therefore admits a non-trivial null space, enabling exact cancellation of dominant-path interference.
\end{remark}

Reliable decoding requires accounting
for the full superposition of multipath components captured by the effective
channel model $\mathbf{G}_{n,m}$ presented in \eqref{InterferenceModel}. To this end, we define the aggregate interference
matrix associated with stream $n$ as $\mathbf{H}^{\mathrm{int}}_{n,m} \in \mathbb{C}^{K \times (N_T-1)K_s}$, given by:
$\mathbf{H}^{\mathrm{int}}_{n,m}
=
\big[\,
\mathbf{G}_{1,m},\ldots,\mathbf{G}_{n-1,m},
\mathbf{G}_{n+1,m},\ldots,\mathbf{G}_{N_T,m}
\,\big].$

{\color{black}
To separate the streams, the design of the linear combiner $\mathbf{A}_{n,m}$ is cast as a constrained Zero-Forcing (ZF) subspace receiver. Specifically, $\mathbf{A}_{n,m}$ is constructed to exactly suppress the interference by projecting the desired effective channel onto the orthogonal complement of the interference subspace:
\begin{equation}\label{eq:ZF_combiner}
\mathbf{A}_{n,m}^{\mathsf H}
=
\left(
\mathbf{G}_{n,m}^{\mathsf H}
\boldsymbol{\Pi}_{\mathbf{H}^{\mathrm{int}}_{n,m}}^{\perp}
\mathbf{G}_{n,m}
\right)^{-1}
\mathbf{G}_{n,m}^{\mathsf H}
\boldsymbol{\Pi}_{\mathbf{H}^{\mathrm{int}}_{n,m}}^{\perp},
\end{equation}
where $\boldsymbol{\Pi}_{\mathbf{H}^{\mathrm{int}}_{n,m}}^{\perp} \in \mathbb{C}^{K \times K}$ is the orthogonal projection onto the left null space of $\mathbf{H}^{\mathrm{int}}_{n,m}$ (i.e., the null space of $(\mathbf{H}^{\mathrm{int}}_{n,m})^{\mathsf{H}}$). By construction, the linear combiner in \eqref{eq:ZF_combiner} effectively suppresses the IAI component in the decoded subspace \eqref{eq:alphan_decomposed} arising from multipath-induced orthogonality leakage.}

{\color{black}
While this closed-form ZF approach guarantees deterministic processing latency and effective stream separation, it is subject to the standard ZF trade-off: in environments where multipath components are highly correlated or the effective subspaces are poorly conditioned, the projection may result in noise enhancement. While spatial parameter estimation inherently exploits the full receive array ($M_C$ antennas) to resolve paths, the proposed equalizer resolves interference in the code/delay domain. Consequently, subsequent symbol equalization and decoding are performed on a single receive antenna, although spatial diversity combining (e.g., Maximum Ratio Combining) could optionally be applied across the array.
}

The equalizer derived under this framework relies on the estimated $\mathbf{G}_{n,m}$, establishing a direct link between multipath
parameter estimation and data decoding performance.

\begin{remark}
{The proposed receiver design relies exclusively on linear constrained optimization and admits closed-form solutions, resulting in deterministic processing latency and hardware-friendly implementation. Unlike iterative or non-linear receivers, the proposed approach preserves low complexity while explicitly exploiting multipath structure for interference mitigation.}

\end{remark}

\subsection{Pilot-Aided Multipath Parameter Estimation}
\label{subsec:multipath_estimation}

\begin{algorithm}[!t]
\footnotesize
\caption{Multipath Parameters Estimation}
\label{alg:channel_est}
\begin{algorithmic}[1]
\Require Received signal $\mathbf{Y}$, Pilot $\mathbf{u}$, Array manifold $\mathbf{a}(\cdot)$
\Ensure Parameter set $\Psi = \{\hat{\tau}_p, \hat{\phi}_p, \hat{\nu}_p, \hat{\gamma}_p\}_{p=0}^{P-1}$

\Statex \textit{Spatial Estimation (Subspace-Based):}
\State $\hat{\mathbf{R}} \leftarrow \frac{1}{K}\mathbf{Y}^{\mathsf{H}}\mathbf{Y}$
\State $\mathbf{E}_n \leftarrow \text{NoiseSubspace}(\hat{\mathbf{R}})$
\State $\{\hat{\phi}_p\}_{p=0}^{{P}-1} \leftarrow
\underset{\phi}{\text{find\_peaks}}
\left(
\frac{1}{\mathbf{a}(\phi)^{\mathsf{H}} \mathbf{E}_n \mathbf{E}_n^{\mathsf{H}} \mathbf{a}(\phi)}
\right)$

\For{$p = 0, \dots, P-1$}
    \Statex \textit{Beamforming \& Delay Estimation:}
    \State Form spatial filter $\mathbf{w}_p \leftarrow \mathbf{a}(\hat{\phi}_p)$
    \State $z_p[k] \leftarrow \mathbf{w}_p^{\mathsf{H}} \mathbf{y}_k$
    \State $\hat{\tau}_p \leftarrow
    $ as in \eqref{eq:tau_p}
    \Statex \textit{Doppler Estimation (ML Search):}
    \State Construct delay--angle compensated template
    \State $\hat{\nu}_p \leftarrow$ as in \eqref{eq:vp}
\EndFor

\Statex \textit{Gain Recovery (Least Squares):}
\State $\mathbf{y} \leftarrow \mathrm{vec}(\mathbf{Y})$
\State Construct LS design matrix 
$\mathbf{H}_{\mathrm{LS}} = [\mathbf{h}_0, \dots, \mathbf{h}_{P-1}]$
\State $\mathbf{h}_p$ as in \eqref{eq:LS}
\State $\hat{\boldsymbol{\gamma}} \leftarrow
(\mathbf{H}_{\mathrm{LS}}^{\mathsf{H}} \mathbf{H}_{\mathrm{LS}})^{-1}
\mathbf{H}_{\mathrm{LS}}^{\mathsf{H}} \mathbf{y}$
\State \Return $\Psi$
\end{algorithmic}
\end{algorithm}

A known pilot sequence $\mathbf{u}\in\mathbb{C}^{K_u}$, with $K_u < K_s$, is embedded
into the first transmit waveform prior to the Dual-Orthogonality construction,
requiring:
\begin{equation} \label{eq:pilot-consistency}
\left[ \mathbf{s}_1 \right]_{1:K_u} = \mathbf{u}
\quad \Leftrightarrow \quad
\mathbf{C}_u \boldsymbol{\alpha}_1 = \mathbf{u}
\end{equation}
where $\mathbf{C}_u \in \mathbb{C}^{K_u \times K_s}$ consists of the first $K_u$ rows
of $\mathbf{C}_1$. 
{\color{black}
The pilot is embedded within the existing waveform block and therefore
does not require an additional time interval, frequency allocation, or
transmit subspace. However, allocating these fixed pilot samples consumes a portion of the available signaling dimensions, thereby introducing a payload overhead. A detailed analysis of this overhead, together with the constructive solution satisfying \eqref{eq:pilot-consistency}, is provided in Appendix~\ref{app:pilot_consistency}.
}
This pilot sequence enables estimation of the multipath parameters required to form the effective channel.

The channel estimation pipeline first identifies the spatial signatures. From the
received block $\mathbf{Y}$, the sample covariance matrix
$\hat{\mathbf{R}} = \frac{1}{K}\mathbf{Y}^{\mathsf{H}}\mathbf{Y}$ is formed, and the angles of
arrival are estimated by identifying the $P$ dominant peaks of the
subspace-based spatial spectrum as
\begin{equation}
\hat{\phi}_p
=
\arg\max_{\phi}
\frac{1}{\mathbf{a}^{\mathsf{H}}_{\mathrm{rx}}(\phi)
\mathbf{E}_n \mathbf{E}_n^{\mathsf{H}}
\mathbf{a}_{\mathrm{rx}}(\phi)},
\end{equation}
for $p = 0,\dots,P-1$ and where $\mathbf{E}_n$ denotes the noise subspace obtained from the eigenvalue decomposition of $\hat{\mathbf{R}}$, by collecting the eigenvectors corresponding
to the smallest eigenvalues, under the standard assumption that the number of
dominant multipath components is smaller than the receive array dimension.
{\color{black} Here, $P$ denotes the number of dominant, resolvable multipath components. While practical systems can estimate $P$ via MUSIC peak thresholding or standard information-theoretic criteria, the estimation procedure in Algorithm~\ref{alg:channel_est} and the numerical evaluations in Section~\ref{NumericalSim} assume $P$ is known a priori. This strictly isolates the interference-cancellation performance of the proposed equalizer from model-order estimation errors.}
Conditioned on these angles, the received signal is spatially filtered 
to emphasize individual path contributions.\footnote{
To reduce complexity, angle and delay are estimated sequentially instead of via a joint two-dimensional search, assuming spatial filtering sufficiently suppresses residual multipath components. The impact of residual leakage is evaluated in simulations for an $8\times1$ receive array, consistent with~\cite{keskin2021jstsp}
}  Specifically, for each estimated angle $\hat{\phi}_p$, 
$\mathbf{a}_{\mathrm{rx}}(\hat{\phi}_p)$ is applied to $\mathbf{Y}$ as a beamforming vector.

For equalization purposes, it is sufficient to work with a
stream-dependent effective path gain. Specifically, the complex
coefficient associated with each propagation path can be
defined as $\gamma_{p,n} \triangleq \gamma_p
[\mathbf{a}_{\mathrm{tx}}(\theta_p)]_n$,
which directly parameterizes the contribution of path $p$ to stream $n$. 
{\color{black} Because the pilot sequence is embedded exclusively in the first
transmit stream to minimize overhead (Appendix~\ref{app:pilot_consistency}),
the remote receiver cannot explicitly resolve the Angle of Departure (AoD, $\theta_p$).
Consequently, the receiver employs a far-field approximation, absorbing the unresolved
spatial phase into the common estimated macroscopic gain such that $\gamma_{p,n} \approx \hat{\gamma}_p$
for all $n$. While this introduces a residual spatial phase error, the constrained
Zero-Forcing projection in \eqref{eq:ZF_combiner} remains robust. The interference subspace
is primarily anchored by the known temporal codes ($\mathbf{C}_\ell$) and the estimated
delay--Doppler operators ($\mathbf{H}_p$), allowing the equalizer to effectively suppress
dominant interference by targeting its temporal support rather than relying on strict spatial
phase coherence.}
This equivalent representation preserves the structure of $\mathbf{G}_{n,m}$ while simplifying
the receiver-side processing.

Let $\mathbf{y}_k \in \mathbb{C}^{M_C\times 1}$ denote the $k$-th row of $\mathbf{Y}$, 
i.e., the spatial receive vector at time index $k$. 
Delay estimation is performed by correlating the beamformed signal 
with the known pilot sequence $\mathbf{u}$, yielding
\begin{equation} \label{eq:tau_p}
    \hat{\tau}_p
    =
    {\arg \max_{\tau}}
    \Big|
    \sum_{k=0}^{K_u-1}
    \left( \mathbf{a}_{\mathrm{rx}}^{\mathsf{H}}(\hat{\phi}_p)\mathbf{y}_k \right)
    u^*(k-\tau)
    \Big| .
\end{equation}
To resolve fractional delays with sub-sample precision, parabolic interpolation is then applied to the correlation magnitude of the peak and its immediate adjacent samples.
Subsequently, the Doppler shifts $\hat{\nu}_p$ (directly related to the relative velocities
$\hat{v}_p$) are estimated via a Maximum Likelihood (ML) grid search. 
While AoA estimation exploits the spatial covariance structure of the full received block, 
delay and Doppler estimation rely on the known pilot structure to enable coherent 
projection-based estimation in the temporal domain. 
This involves maximizing the normalized coherent projection of the observation onto the
delay-and-angle-compensated pilot template $\mathbf{h}$:
\begin{equation} \label{eq:vp}
    \hat{\nu}_p =
    \arg \max_\nu \,
    \frac{\left| \mathbf{h}^{\mathsf{H}}(\hat{\tau}_p, \hat{\phi}_p, \nu) \mathbf{y} \right|^2}
    {\| \mathbf{h}(\hat{\tau}_p, \hat{\phi}_p, \nu) \|^2} .
\end{equation}
where $\mathbf{y}\triangleq \mathrm{vec}(\mathbf{Y})\in\mathbb{C}^{K M_C \times 1}$, and
$\mathbf{h}(\hat{\tau}_p, \hat{\phi}_p, \nu)
=
\mathbf{a}_{\mathrm{rx}}(\hat{\phi}_p) \otimes
\mathbf{D}(\nu)\mathbf{F}^{\mathsf{H}}\mathbf{B}(\hat{\tau}_p)\mathbf{F}\bar{\mathbf{u}}$,
with $\bar{\mathbf{u}} = [\mathbf{u}^{\top}, \mathbf{0}_{1 \times (K-K_u)}]^{\top}$ denoting the zero-padded pilot. This properly incorporates the estimated arrival angle, delay, and Doppler phase rotation.

Finally, the complex path gains $\{\hat{\gamma}_p\}$ are recovered via linear
least-squares (LS). We define the LS design matrix
$\mathbf{H}_{\mathrm{LS}}=[\mathbf{h}_{0},\ldots, \mathbf{h}_{P-1}]$, where
\begin{equation} \label{eq:LS}
\mathbf{h}_p
=
\mathbf{a}_{\mathrm{rx}}(\hat{\phi}_p)\otimes
\mathbf{D}(\hat{\nu}_p)\mathbf{F}^{\mathsf{H}}\mathbf{B}(\hat{\tau}_p)\mathbf{F}\,\bar{\mathbf{u}}.
\end{equation}
The LS estimate of the path gains is then given by
\begin{equation}
\hat{\boldsymbol{\gamma}}
=
(\mathbf{H}_{\mathrm{LS}}^{\mathsf{H}}\mathbf{H}_{\mathrm{LS}})^{-1}
\mathbf{H}_{\mathrm{LS}}^{\mathsf{H}} \mathbf{y}.
\end{equation}
The estimation procedure is summarized in Algorithm~\ref{alg:channel_est}.

{\color{black} It is important to note that the proposed parameter estimation procedure operates sequentially (spatial angle $\to$ delay $\to$ Doppler) rather than jointly. While a joint multi-dimensional ML search over $(\phi, \tau, \nu)$ would yield theoretically optimal estimates, it is computationally prohibitive for practical real-time ISAC systems. The sequential approach trades strict optimality for computational feasibility. Because the delay and Doppler estimations are conditioned on the initial angle estimates, this architecture is theoretically susceptible to error propagation. However, when the receive array size is sufficiently large to spatially resolve the dominant paths, the spatial filtering step (beamforming) effectively decouples the multipath components. This isolates the delay-Doppler estimation for each path, thereby bounding the residual error propagation and allowing the sequential estimator to approach the Cram\'er-Rao Bound (CRB) at high SNR regime, as corroborated by the numerical results in Section \ref{NumericalSim}.

}

\subsection{Theoretical Performance Limits}
\label{subsec:param_est_perf} To quantify the fundamental limits of the proposed channel estimation framework, we derive a global Cramér-Rao Bound (CRB) for the effective subspace channel matrix $\mathbf{G}_{n,m}$ in \eqref{InterferenceModel}. We define the global parameter vector $\boldsymbol{\psi} = [\boldsymbol{\psi}_0^\top, \dots, \boldsymbol{\psi}_{P-1}^\top]^\top$, where the parameters for the $p$-th path ($p=0,\dots,P-1$) are given by $\boldsymbol{\psi}_p = [\tau_p, \phi_p, \nu_p, \Re\{\gamma_p\}, \Im\{\gamma_p\}]^\top$. $\gamma_p \in \mathbb{C}$ denotes the complex path gain, decomposed into its real and imaginary parts. The Fisher Information Matrix (FIM) for $\boldsymbol{\psi}$, denoted as $\mathbf{J}_{\boldsymbol{\psi}}$, is computed via the Slepian-Bangs formula based on the signal model in \eqref{final_model}.
However, since communication reliability depends on the accuracy of the effective channel reconstruction rather than individual physical parameters, we bound the mean squared error of the matrix $\mathbf{G}_{n,m}$ itself. Let $\mathbf{g}(\boldsymbol{\psi}) = \text{vec}(\mathbf{G}_{n,m})$ represent the vectorized effective channel as a function of the physical parameters. The Jacobian matrix of this transformation is $\mathbf{T} = \partial \mathbf{g} / \partial \boldsymbol{\psi}$.

The global lower bound on the variance of the channel estimator is then given by:
\begin{equation}
    \mathbb{E}\left[\|\hat{\mathbf{G}}_{n,m} - \mathbf{G}_{n,m}\|_F^2\right] \geq \text{Tr}\left( \mathbf{T} \mathbf{J}_{\boldsymbol{\psi}}^{-1} \mathbf{T}^{\mathsf{H}} \right).
    \label{eq:global_crb}
\end{equation}
This bound captures the coupled impact of delay, Doppler, angle, and complex gain errors on the final subspace projection, providing a direct benchmark for the estimation performance reported in Section \ref{NumericalSim}.

\section{SAR Imaging Processing}
\label{sec:env_reconstruction}

This section focuses on imaging processing at the ISAC node using the reflected echoes of the transmitted Dual-Orthogonality waveforms. Unlike the remote communication receiver, which estimates channel parameters for data
decoding, the sensing receiver can exploit mobility to synthesize a large aperture and perform high-resolution imaging. Crucially, the proposed waveform
enables coherent multi-antenna Synthetic Aperture Radar (SAR) processing without requiring frequency- or time-division across transmitters, allowing all antennas to contribute to imaging at full bandwidth and full aperture \cite{manzoni2023tgrs}.

\subsection{Sensing Observation Model and Virtual Array Structure}

In the monostatic sensing configuration, the SAR observation corresponds to a particular instance of the propagation model in Sec.~III. 
Specifically, at aperture position of the sensing receiver $\mathbf{x}_{q}$, the received baseband signal is obtained from \eqref{comm_rec_equation} by considering a co-located transmit and sensing array, where the channel captures both direct target reflections and higher-order multipath echoes.
The continuous-time echo can be written as
\begin{equation} \label{eq:sensing_signal}
\hspace{-0.1cm}\mathbf{y}_{q}(t)
\hspace{-0.1cm}=
\hspace{-0.15cm}\sum_{p=0}^{P-1}
\hspace{-0.15cm}\gamma_{p,q}
e^{j2\pi \nu_{p,q} t}
\mathbf{a}_{\mathrm{rx}}(\phi_{p,q})
\mathbf{a}_{\mathrm{tx}}^{\top}(\theta_{p,q})
\mathbf{s}\big(t-\tau_{p,q}\big)
\hspace{-0.1cm}+
\hspace{-0.1cm}\mathbf{w}_{q}(t),
\end{equation}
where $\Psi_{q} = \{ \tau_{p,q}, \theta_{p,q}, \phi_{p,q}, \gamma_{p,q}, \nu_{p,q} \}_{p=0}^{P-1}$ represents the set of multipath parameters associated with aperture position $\mathbf{x}_{q}$.
Matched filtering (pulse compression) \cite{manzoni2023tgrs} is then performed to separate the spatial channels. The observation at the $m$-th receive antenna $y_{q,m}(t)$, consisting in the $m$-th column of (\ref{eq:sensing_signal}), is cross-correlated with the waveform $s_{n}(t)$ transmitted by the $n$-th antenna. {The resulting pulse-compressed complex envelope for the $(n,m)$ transmit-receive pair is given by
\begin{equation}\label{pulsec}
    s_{q,n,m}(t) = \int y_{q,m}(t+\tau) s_{n}^*(\tau) \, d\tau,
\end{equation}
whose discrete-time sampled version is $s_{q,n,m}[k] = s_{q,n,m}(kT_s)$.} By combining the outputs across all $N_T$ transmit and $M_S$ receive antennas, a virtual array with $N_T M_S$ elements is synthesized.
Under half-wavelength inter-element spacing at the resulting virtual elements positions, given by proper array geometry, the virtual array enlarges the effective aperture and enables grating-lobe-free operation while improving angular resolution and discrimination of closely spaced scatterers.

Unlike OFDM-based multi-transmit SAR schemes that rely on FDMA or subcarrier interleaving to preserve orthogonality, thereby reducing per-transmitter bandwidth or introducing periodic sidelobes, the Dual-Orthogonality waveform preserves orthogonality across transmit antennas within a prescribed delay--Doppler region without bandwidth partitioning, as described in Sec.~III-A. As a result, SAR processing can coherently exploit the full signal bandwidth across all transmit antennas.

\subsection{Time-Domain Back-Projection Imaging}

\begin{algorithm}[!t]
\footnotesize
\caption{MIMO-SAR Image Formation}
\label{alg:TDBP_compact}
\begin{algorithmic}[1]
\Require Range-compressed data $\{s_{q,n,m}[k]\}$, positions $\{\mathbf{x}_q\}$, phase centers $\{\mathbf{x}_n^{\rm TX},\mathbf{x}_m^{\rm RX}\}$, grid $\mathcal{G}$
\Ensure Reflectivity map $I(\mathbf{r})$
\ForAll{$\mathbf{r}\in\mathcal{G}$}
  \State $I(\mathbf{r}) \gets 0$
  \ForAll{$(q,n,m)$}
    \State $\tau \gets \tau_{q,n,m}(\mathbf{r})$
    \State $v \gets \mathrm{interp}\!\left(s_{q,n,m},\tau\right)$
    \State $I(\mathbf{r}) \gets I(\mathbf{r}) + v\,e^{j2\pi f_c \tau}$
  \EndFor
\EndFor
\State $I(\mathbf{r}) \gets {|I(\mathbf{r})|}/{\max_{\mathbf{r} \in \mathcal{G}} |I(\mathbf{r})|}$
\end{algorithmic}
\end{algorithm}

Imaging processing is performed at the ISAC node, which repeatedly transmits waveform blocks along the synthetic aperture. Let us consider the aperture positions $\mathbf{x}_q$ with $q \in \{1,\dots,Q\}$.
At each aperture position, the received echo follows the monostatic propagation model, obtaining (\ref{pulsec}).

Under the standard SAR assumption of negligible intra-pulse Doppler and coherent phase evolution across aperture positions, residual phase evolution is compensated by the exponential term in the back-projection step.
For a candidate image location $\mathbf{r}$, the two-way propagation delay associated with aperture position $q$ and transmit–receive pair $(n,m)$ is
$\tau_{q,n,m}(\mathbf{r})
=
{
\|\mathbf{x}_q + \mathbf{x}^{\mathrm{TX}}_n - \mathbf{r}\|
+
\|\mathbf{x}_q + \mathbf{x}^{\mathrm{RX}}_m - \mathbf{r}\|
}/{c}$,
where $\mathbf{x}^{\mathrm{TX}}_n$ and $\mathbf{x}^{\mathrm{RX}}_m$ denote local phase-center coordinates relative to the platform reference frame.
The MIMO-SAR time-domain back-projection (TDBP) operator is then given by
\begin{equation}\label{eq:TDBP_unnorm}
\tilde{I}(\mathbf{r})
=
\sum_{q=1}^{Q}
\sum_{n=1}^{N_T}
\sum_{m=1}^{M_S}
s_{q,n,m}\!\left(\tau_{q,n,m}(\mathbf{r})\right)
\, e^{j 2\pi f_c \tau_{q,n,m}(\mathbf{r})}.
\end{equation}
The reflectivity map is obtained by coherent accumulation across time acquisitions (synthetic aperture) and spatial MIMO channels after interpolation of the range-compressed signals at the geometrically consistent delays  \cite{manzoni2023tgrs}.
 To ensure that
image intensity is independent of the number of transmitted pulses and path-loss
effects, the final SAR image is normalized as
\begin{equation}
    I(\mathbf{r})
    =
    \frac{|\tilde{I}(\mathbf{r})|}
    {\max_{\mathbf{r}} |\tilde{I}(\mathbf{r})|}.
    \label{eq:TDBP_norm}
\end{equation}
Algorithm~\ref{alg:TDBP_compact} summarizes the implemented TDBP procedure.

{\color{black}
\subsection{Orthogonality under Delay--Doppler Perturbations}
A key property of the proposed waveform for SAR imaging is its exact
delay-domain orthogonality at the discrete design points, while general
delay--Doppler propagation may introduce residual cross-ambiguity.
} 
Under the sampled narrowband model introduced in Sec.~III, each propagation path is represented in discrete time by
\begin{equation}
\mathbf{H}(\tau,\nu)
=
\mathbf{D}(\nu)\mathbf{F}^{\mathsf{H}}\mathbf{B}(\tau)\mathbf{F},
\end{equation}
which coincides with the delay--Doppler transformation appearing in \eqref{final_model}. In the monostatic SAR configuration, this operator models the delay--Doppler transformation associated with each reflector.

\begin{proposition}
Let $\mathbf{s}_i$ and $\mathbf{s}_\ell$ denote the waveforms transmitted
by antennas $i$ and $\ell$, respectively. Let
$\mathcal{A}_d = \{\tau_k = k T_s\}_{k=0}^{K_z-1}$
denote the discrete set of admissible delays enforced during the waveform
construction.
{\color{black}
Let $\mathbf{T}_k\in\mathbb{C}^{K\times K}$ denote the linear
zero-padded integer-delay operator associated with
$\tau_k=kT_s$. For every constrained waveform pair $i<\ell$, the
Dual-Orthogonality construction guarantees
}
\begin{equation}
{\color{black}
\mathbf{s}^{\mathsf{H}}_i\mathbf{T}_k\mathbf{s}_\ell=0,
\qquad
k=0,\ldots,K_z-1,\quad i<\ell.
}
\label{eq:exact_integer_delay_orthogonality}
\end{equation}
{\color{black}
Therefore, exact cross-stream orthogonality is guaranteed at the
integer-delay, zero-Doppler design points explicitly included in the
waveform construction. For fractional delays and/or nonzero Doppler
shifts, the corresponding delay--Doppler operator is not included in the
null-space constraints, and the cross-term
$\mathbf{s}^{\mathsf H}_i\mathbf{H}(\tau,\nu)\mathbf{s}_\ell$
is generally nonzero. Its magnitude depends on the specific waveform
realization, fractional-delay offset, Doppler shift, bandwidth, and
signaling duration. These off-grid terms are therefore treated as
residual cross-ambiguity rather than as exact orthogonality conditions.
}
\end{proposition}

\begin{proof}
{\color{black}
For $\ell>i$, the transmitted waveform is
$\mathbf{s}_\ell=\mathbf{C}_\ell\boldsymbol{\alpha}_\ell$, where
$\boldsymbol{\alpha}_\ell$ is selected in the null space of
$\mathbf{Q}_\ell$. Hence,
$\mathbf{Q}_\ell\boldsymbol{\alpha}_\ell=\mathbf{0}$.
Considering the block of $\mathbf{Q}_\ell$ associated with waveform $i$, as defined in (3), gives
}
\begin{equation}
{\color{black}
\widetilde{\mathbf{S}}_i
\mathbf{C}_\ell\boldsymbol{\alpha}_\ell
=
\widetilde{\mathbf{S}}_i\mathbf{s}_\ell
=
\mathbf{0}.
}
\end{equation}
{\color{black}
As introduced in (2), the $k$-th row of
$\widetilde{\mathbf{S}}_i$ evaluates the cross-correlation between
$\mathbf{s}_i$ and the $k$-sample linearly delayed version of
$\mathbf{s}_\ell$. Therefore,
$\mathbf{s}^{\mathsf H}_i\mathbf{T}_k\mathbf{s}_\ell=0$
for every $k=0,\ldots,K_z-1$, which proves
\eqref{eq:exact_integer_delay_orthogonality}.
}
\end{proof}

To highlight the impact of multipath and Doppler in
\eqref{eq:TDBP_unnorm}, we can write the pulse-compressed echo at aperture
position $q$ as
$s_q(t)=\sum_p \gamma_{p,q}\, a_{p,q}(t)$.
Substituting into \eqref{eq:TDBP_unnorm} yields a double sum over $q$ and
$p$.
For a true scatterer at $\mathbf r_0$, the geometric delays
$\tau_q(\mathbf r_0)$ align across $q$, producing constructive coherent
integration. Multipath components that are not geometrically consistent
with $\mathbf r$ generate phase histories that do not align across $q$,
and therefore spread their energy over the image grid.
A residual Doppler error $\Delta\nu$ induces a phase drift
$\Delta\phi_q \approx 2\pi \Delta\nu t_q$. Over an aperture of duration
$T_{\mathrm{ap}}$, controlled defocusing is obtained when
$2\pi|\Delta\nu|T_{\mathrm{ap}}\ll \pi$, i.e.,
$|\Delta\nu|\ll 1/(2T_{\mathrm{ap}})$.
{\color{black}
The coherent back-projection mechanism further improves robustness by reinforcing geometrically consistent phase histories, while components that are inconsistent with a candidate image location do not accumulate coherently across the synthetic aperture. As a result, residual cross-ambiguity effects arising from fractional delays or nonzero Doppler are inherently suppressed by the coherent integration process, even though exact waveform orthogonality is formally guaranteed only at the discrete design points.
}

\begin{small}
\begin{table}[!b]
\centering
\footnotesize
\caption{Simulation Parameters}
\label{tab:sim_params}
\renewcommand{\arraystretch}{1.05}
\resizebox{\columnwidth}{!}{%
\begin{tabular}{l c c}
\toprule
\textbf{Parameter} & \textbf{Symbol} & \textbf{Value} \\ 
\midrule
Channel profiles & A / B & Sparse / rich geometric multipath \\
Max delay spread & $\tau_{\max}$ & $\approx 1\,\mu$s (150 m range) \\
Tx velocity & $v$ & [3--30] m/s \\
Bandwidth & $B$ & 40, 200 MHz \\
Carrier frequency & $f_c$ & 28 GHz \\
Pulse duration & $T_p$ & 40 $\mu$s \\
Modulation & -- & QPSK \\
ISAC MIMO array & $N_T \times M_S$ & $8 \times 8$ \\
Comm Rx & $M_C \times 1$ & $8 \times 1$ \\
Active subcarriers & $N_{\text{sub}}$ & {\color{black}1600, 8000} \\
Cyclic prefix & $N_{\text{cp}}$ & {\color{black}51, 255 (1.275 $\mu$s)} \\
\midrule
\multicolumn{3}{c}{\color{black}\textit{Dual-Orthogonality Parameters (for $R_s = 20$ m)}} \\
\color{black}Total samples per block & \color{black}$K$ & \color{black}1600, 8000 \\
\color{black}Samples per stream & \color{black}$K_s$ & \color{black}200, 1000 \\
\color{black}Zero-correlation window & \color{black}$K_z$ & \color{black}7, 28 \\
\color{black}Pilot overhead & \color{black}$K_u$ & \color{black}30 \\
\bottomrule
\end{tabular}%
}
\end{table}
\end{small}

\section{Numerical Simulations}
\label{NumericalSim}

\begin{figure}[t!]
    \centering
    \subfloat[Profile A: Sparse Multipath Scenario\label{fig:profileA}]{%
        \includegraphics[width=0.45\linewidth]{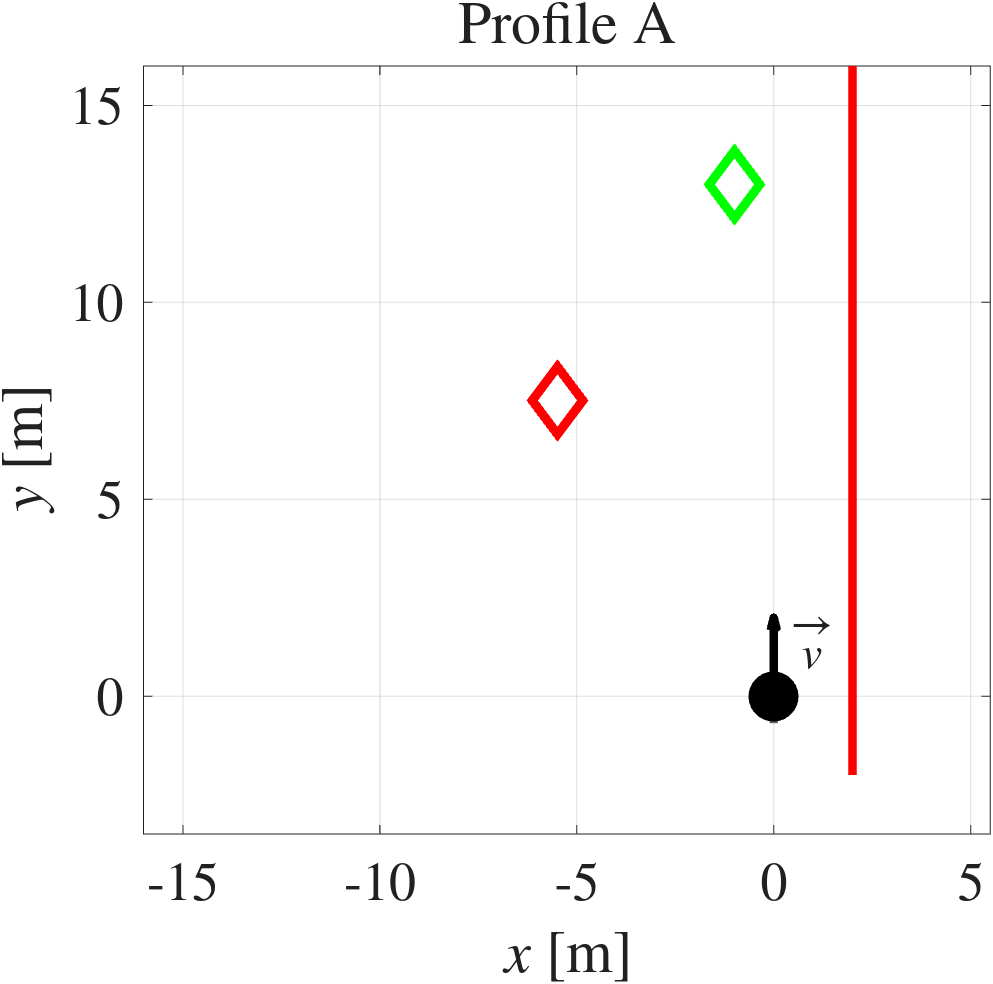}%
    }
    \hfill 
    \subfloat[Profile B: Close spaced Reflectors Multipath Scenario\label{fig:profileB}]{%
        \includegraphics[width=0.46\linewidth]{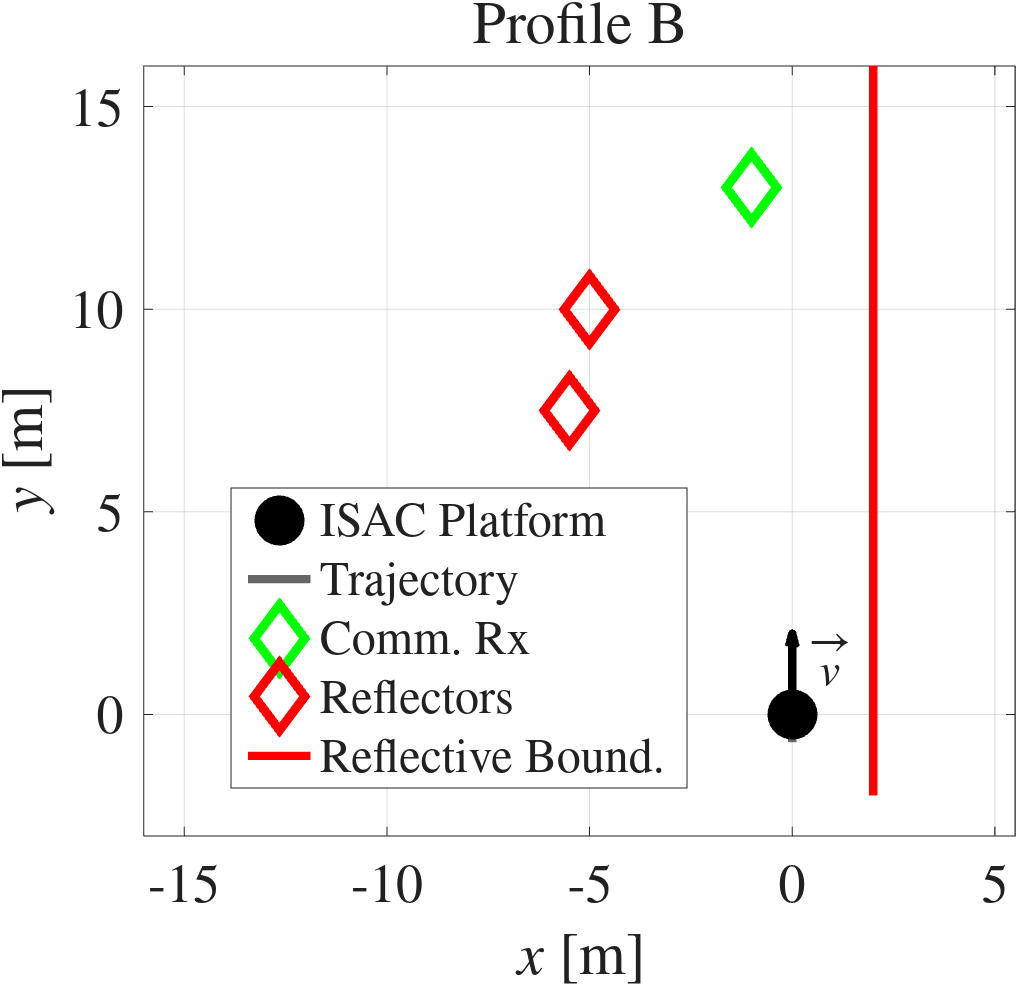}%
    }
    
{\color{black}
\caption{Simulated geometric multipath environments used for numerical
evaluation: (a) Profile~A, characterized by sparse and spatially separated
reflectors; and (b) Profile~B, characterized by closely spaced reflectors
and a richer, more challenging multipath structure.}
\label{fig:sim_profiles}
}
\end{figure}

This section evaluates the proposed Dual-Orthogonality ISAC framework through numerical simulations under dynamic multipath conditions. Communication and sensing performances are analyzed separately and then jointly compared against OFDM-based and \textcolor{black}{OTFS} baselines.

The simulation parameters are summarized in Table~\ref{tab:sim_params}. 
Propagation is modeled using two geometric, ray-based multipath scenarios, 
referred to as \emph{Profile A} and \emph{Profile B}, illustrated in Fig.~\ref{fig:sim_profiles}. 
The geometric construction preserves phase coherence across paths, meaning that 
the complex path phases are determined by the actual propagation distances, 
ensuring physically consistent phase relationships across multipath components 
and aperture positions, which is essential for coherent SAR imaging evaluation. 
Profile~A represents a sparse multipath environment with well-separated reflectors, whereas Profile~B introduces closely spaced reflectors and stronger scattering interactions, yielding a richer and more challenging propagation structure. In both scenarios, the positions of the reflectors and scattering surfaces are randomized across simulation runs while preserving the characteristic spacing constraints of each profile (in particular, the close reflector spacing in Profile~B) and maintaining all elements within the considered field of view. This controlled randomization introduces statistical variability in path phases and gains while preserving the underlying geometric structure of the scene.
To assess robustness under mobility, the transmitting ISAC platform moves with velocities ranging from $3$~m/s to a peak of $30$~m/s, inducing Doppler shifts consistent with typical urban and highway dynamic scenarios.

{\color{black}
The OFDM cyclic prefix $N_{\text{CP}}$ is dimensioned to exceed the maximum communication delay spread ($\approx 1\,\mu\mathrm{s}$). In principle, $\tau_s$ (and thus $K_z$) can be chosen smaller than the OFDM $T_{\text{CP}}$ if a narrower sensing region is targeted. For Dual-Orthogonality waveforms, the nulling window $K_z$ is indeed deliberately dimensioned to exclusively cover the target radar sensing range ($R_s = 20$~m), ensuring physical realizability and large null-space dimensions across all transmit streams as detailed in Section III-A. Any residual long-tail communication multipath arriving from beyond this sensing range is subsequently suppressed by the receiver's linear equalizer.
}
For fair comparison, the OFDM-NeqSI {\color{black} and OTFS} configurations match the Dual-Orthogonality waveform parameters ($B = 40$~MHz, $T_p = 40~\mu$s). {\color{black}The OFDM baselines employ $N_{\text{sub}} = 1600$ and $8000$ active subcarriers, respectively, with a cyclic prefix of $N_{\text{CP}} = 51$ and $255$ samples ($1.275~\mu$s), exceeding the maximum simulated delay spread (150~m range, approximately $1~\mu$s) to ensure interference-free reception under the considered multipath conditions.} {\color{black} The OTFS baseline maps symbols onto a delay-Doppler grid consisting of 16 Doppler bins and 100 delay bins (scaled to $16 \times 500$ for the 200~MHz sensing evaluations), matching the same block lengths.}
For the OFDM baseline, receiver processing employs standard per-antenna LS channel estimation at pilot subcarriers followed by frequency-domain interpolation. To ensure a fair comparison, these independent LS estimates are then coherently combined across the receive array via spatial phase-alignment (beamforming) using the estimated angles of arrival (AoA), thereby suppressing spatial noise and maximizing the effective estimation SNR prior to equalization.

\subsection{Communication Performance}

\begin{figure} [!t]
    \centering
    \includegraphics[width=0.7\linewidth]{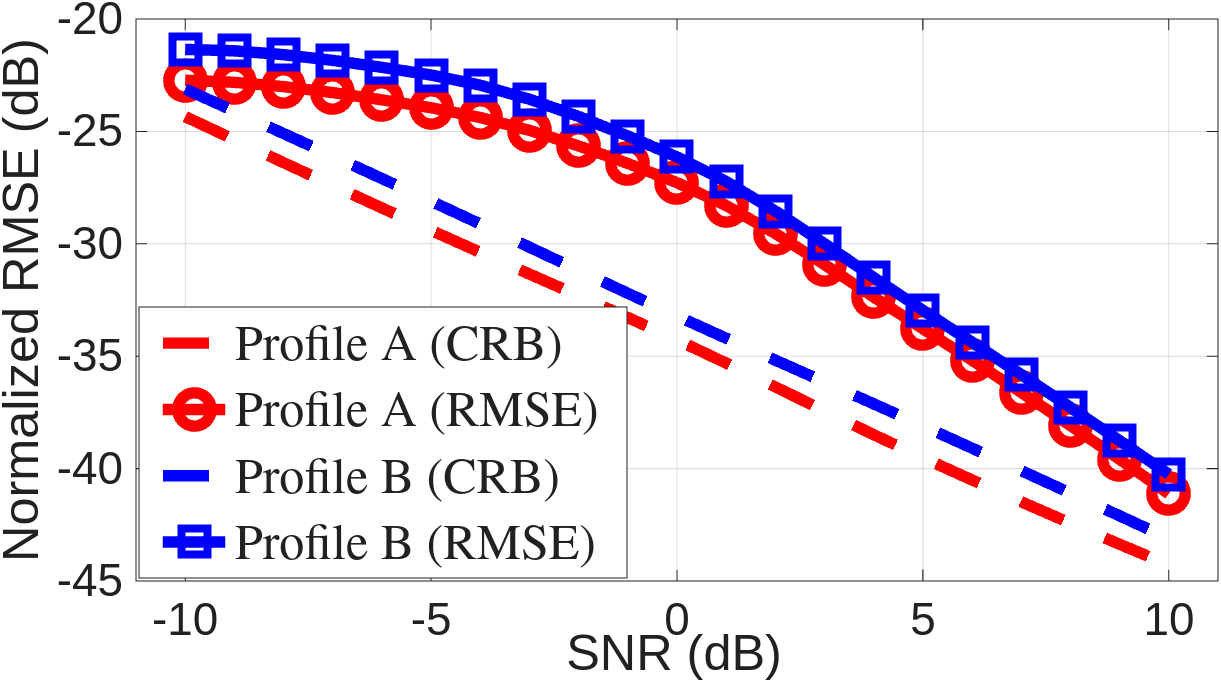}
    \caption{Normalized Mean Squared Error in estimation of $\mathbf{G}_{n,m}$ vs Cramér--Rao Bound for Profile A and Profile B.}
    \label{fig:CRB}
\end{figure}

\begin{figure} [!t]
    \centering
    \includegraphics[width=0.75\linewidth]{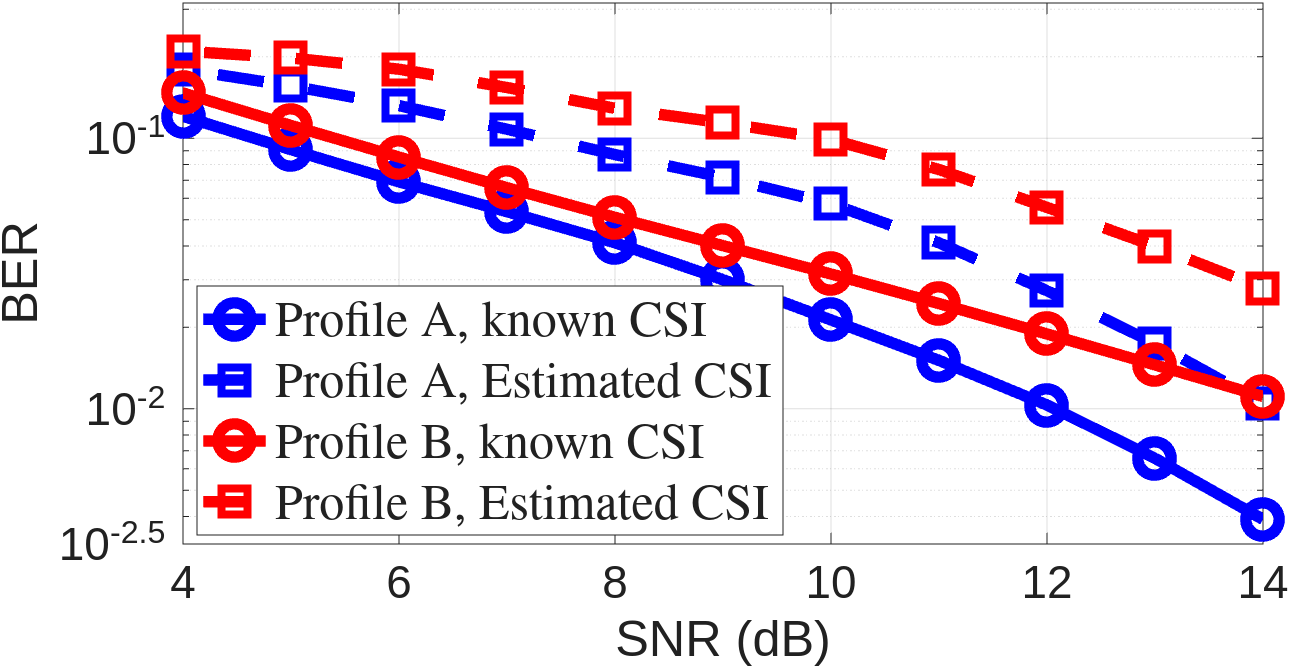}
    \caption{Bit error rate versus SNR, evaluated on a single receive antenna, with $T_p=40\,\mu\text{s}$, and $B=200\,\text{MHz}$.}
    \label{fig:mergedBer}
\end{figure}

\begin{figure*}[!t]
    \centering
    \subfloat[Dual-Orthogonality]{
        \includegraphics[width=0.43\linewidth]{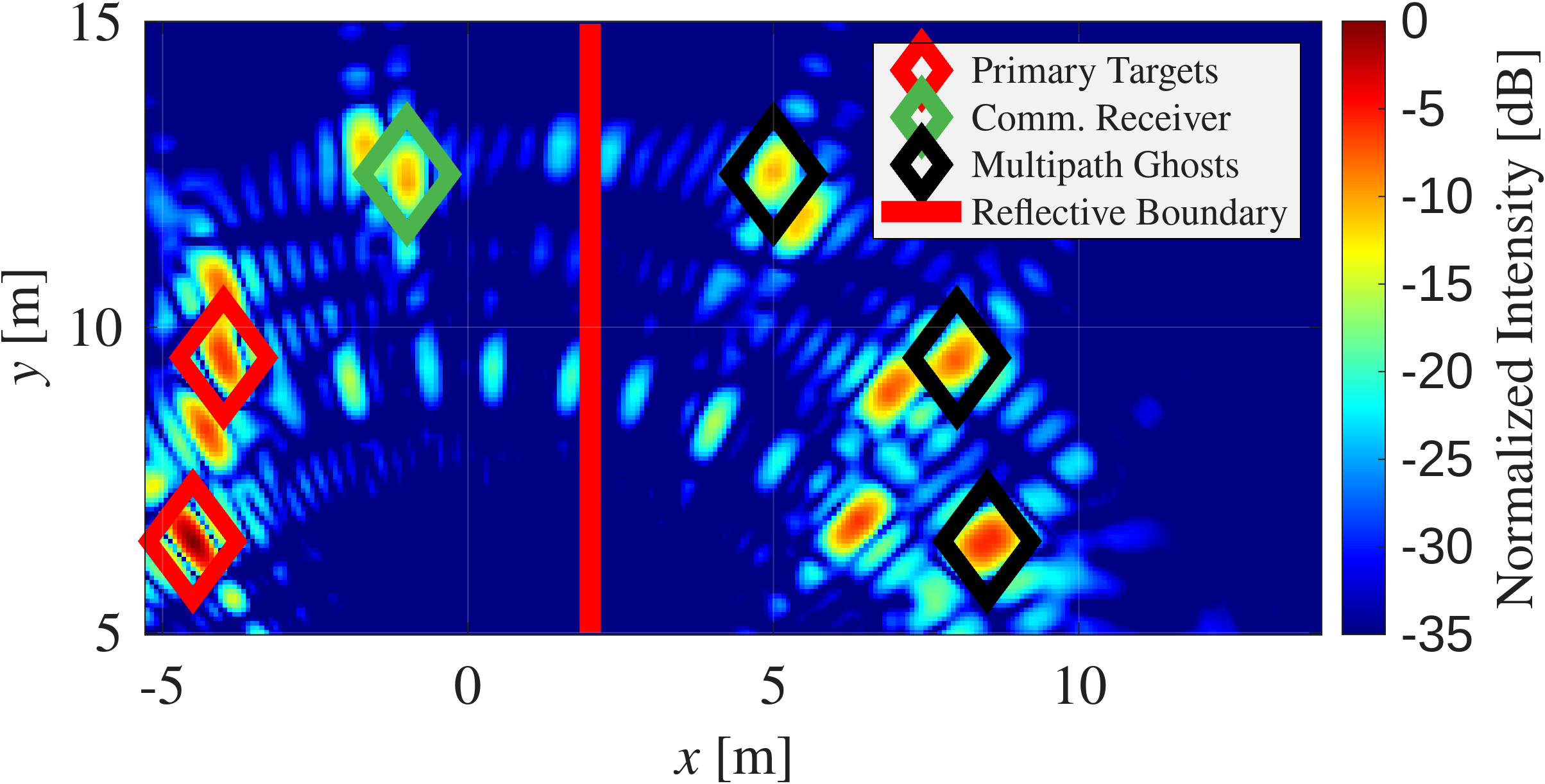}
        \label{subfig:radar_DO}
    }\hfill
    \subfloat[OFDM-FDMA~\cite{cohen2020fdma}]{
        \includegraphics[width=0.43\linewidth]{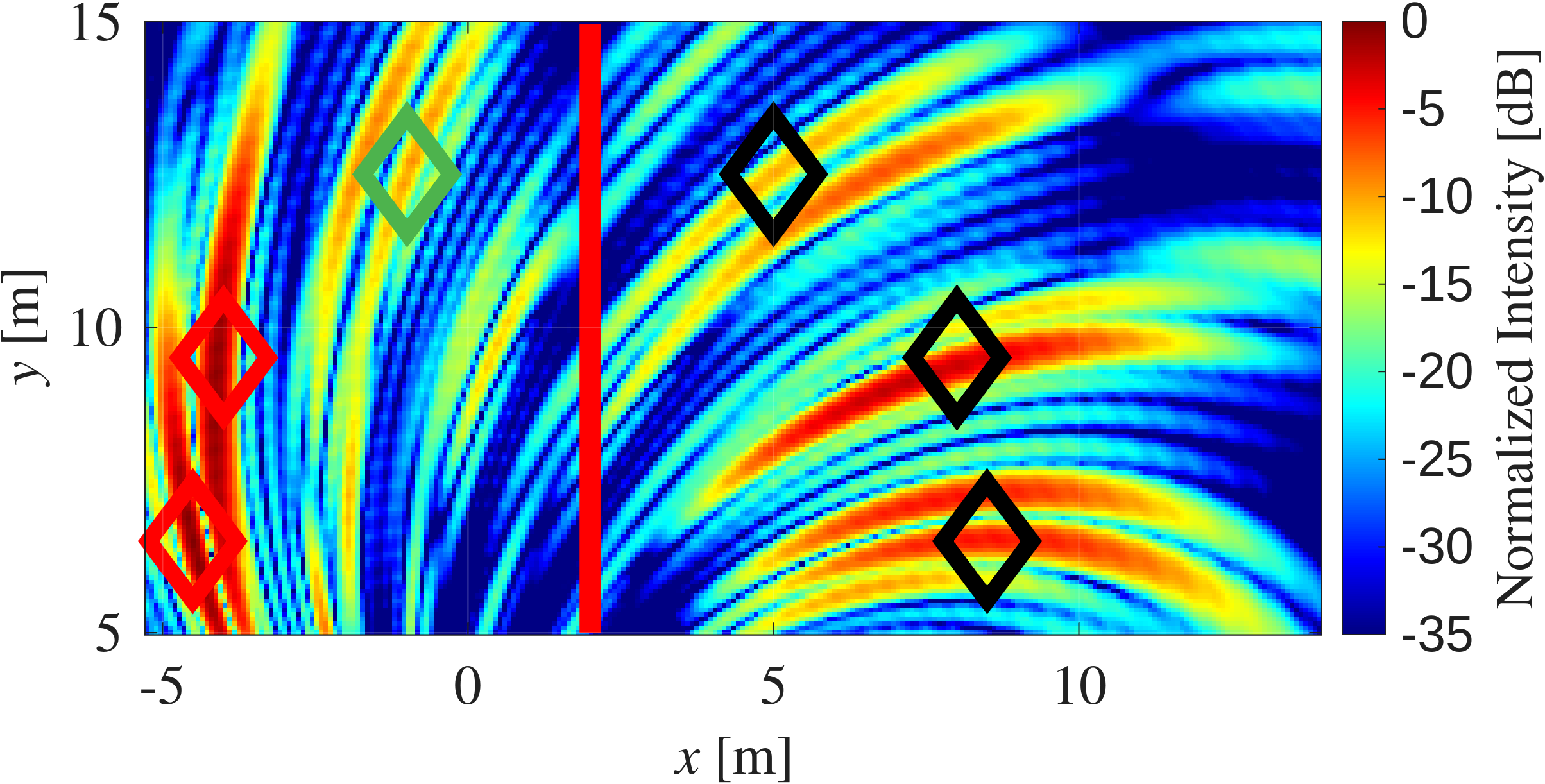}
        \label{subfig:radar_FDMA}
    } \\ 
    \subfloat[OFDM-NeqSI~\cite{Hakobyan2020NeqSI_Optimized}]{
        \includegraphics[width=0.43\linewidth]{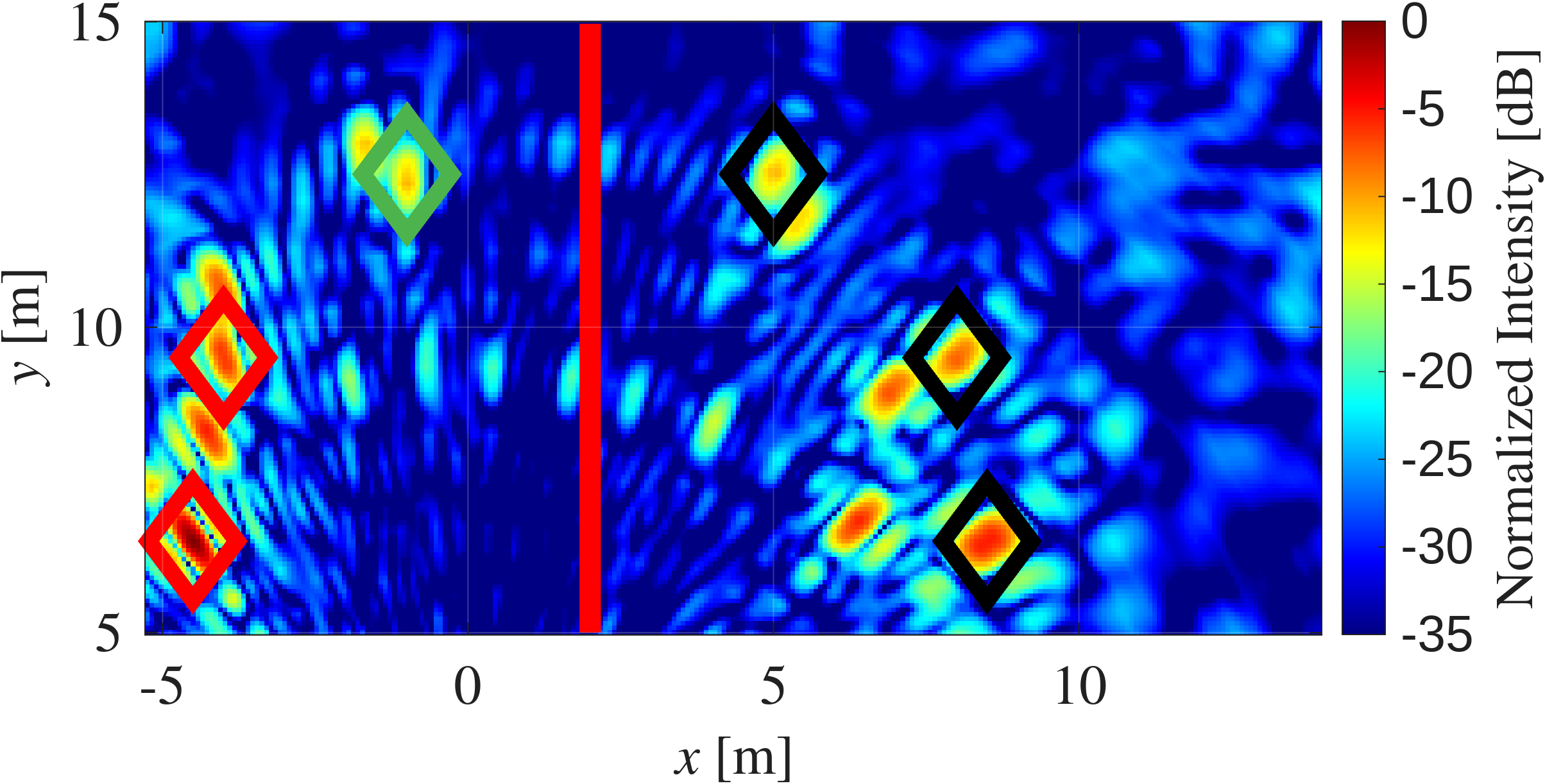}
        \label{subfig:radar_NeqSI}
    }\hfill
    \subfloat[OTFS~\cite{gaudio2020otfsisac}]{
        \includegraphics[width=0.43\linewidth]{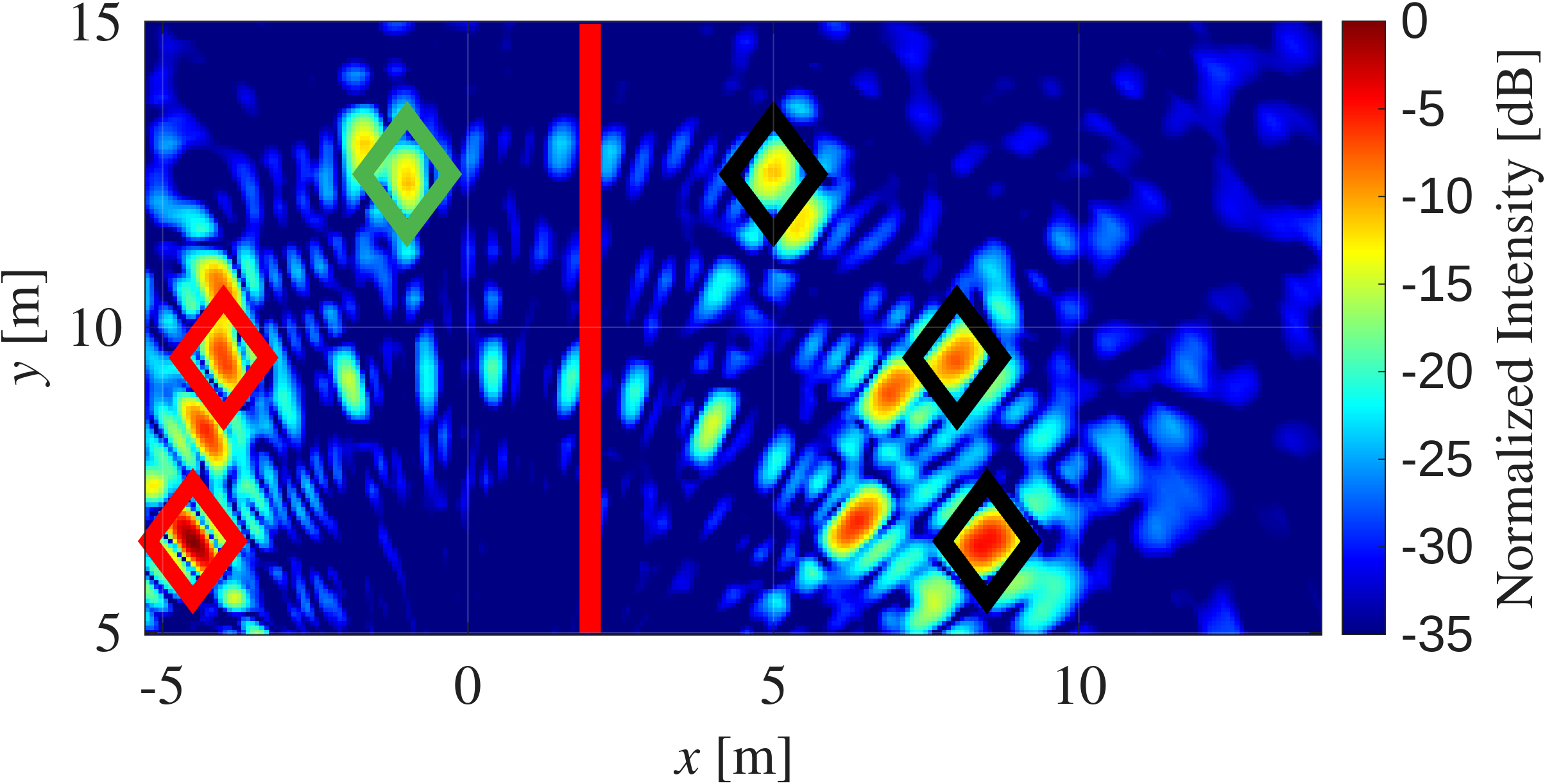}
        \label{subfig:radar_OTFS}
    }
    {\color{black}
    \caption{Pulse-compressed TDBP radar images in Profile~B (Fig.~\ref{fig:profileB}, $B=200$~MHz): (a) proposed Dual-Orthogonality waveform; (b) OFDM-FDMA, showing reduced range resolution from bandwidth partitioning; (c) OFDM-NeqSI, exhibiting elevated sidelobes and structured artifacts; and (d) OTFS, displaying a higher interference floor due to delay--Doppler cross-ambiguity.
    \label{fig:radar_comparison}}
    }
    
\end{figure*}
\begin{figure} [!t]
\centering
    \includegraphics[width=0.75\linewidth]{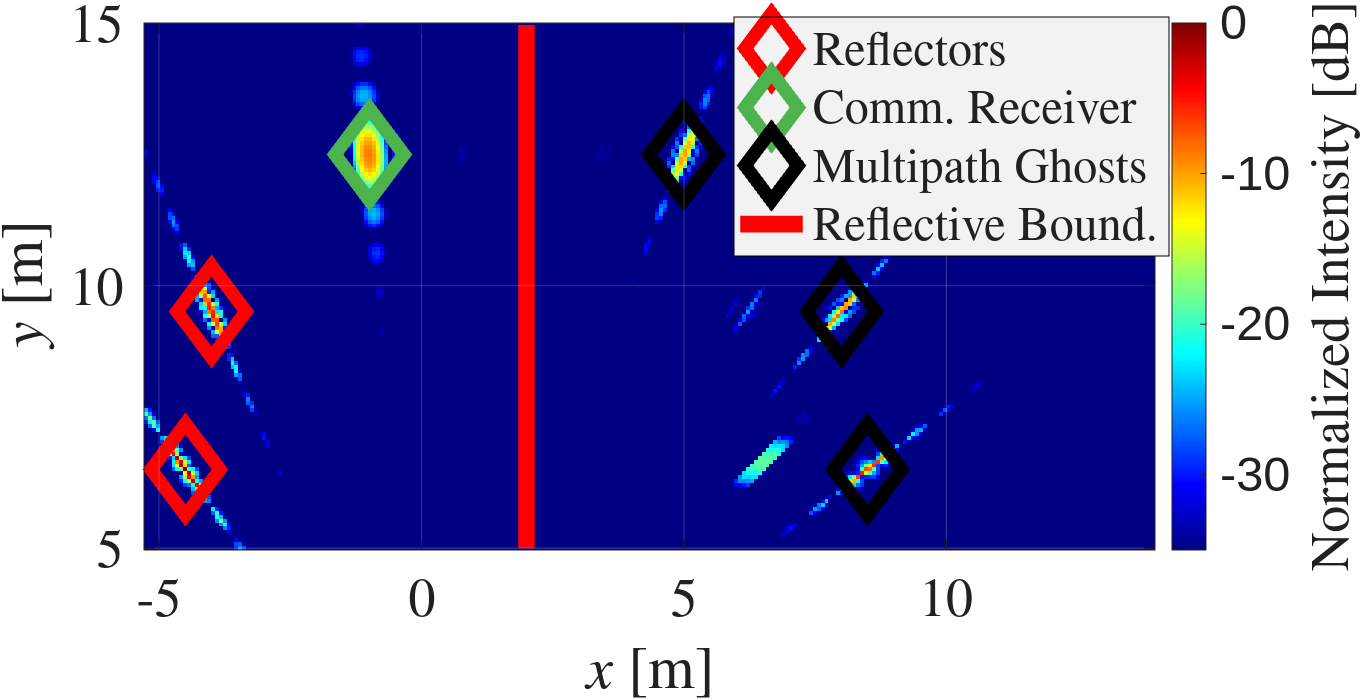}
    \caption{MIMO-SAR image using Dual-Orthogonality waveforms under dynamic multipath conditions ($B=200$~MHz, 1024 SAR acquisitions). }
    \label{fig:SAR}
\end{figure}

\begin{figure}[!ht]
\centering
\includegraphics[width=0.75\linewidth]{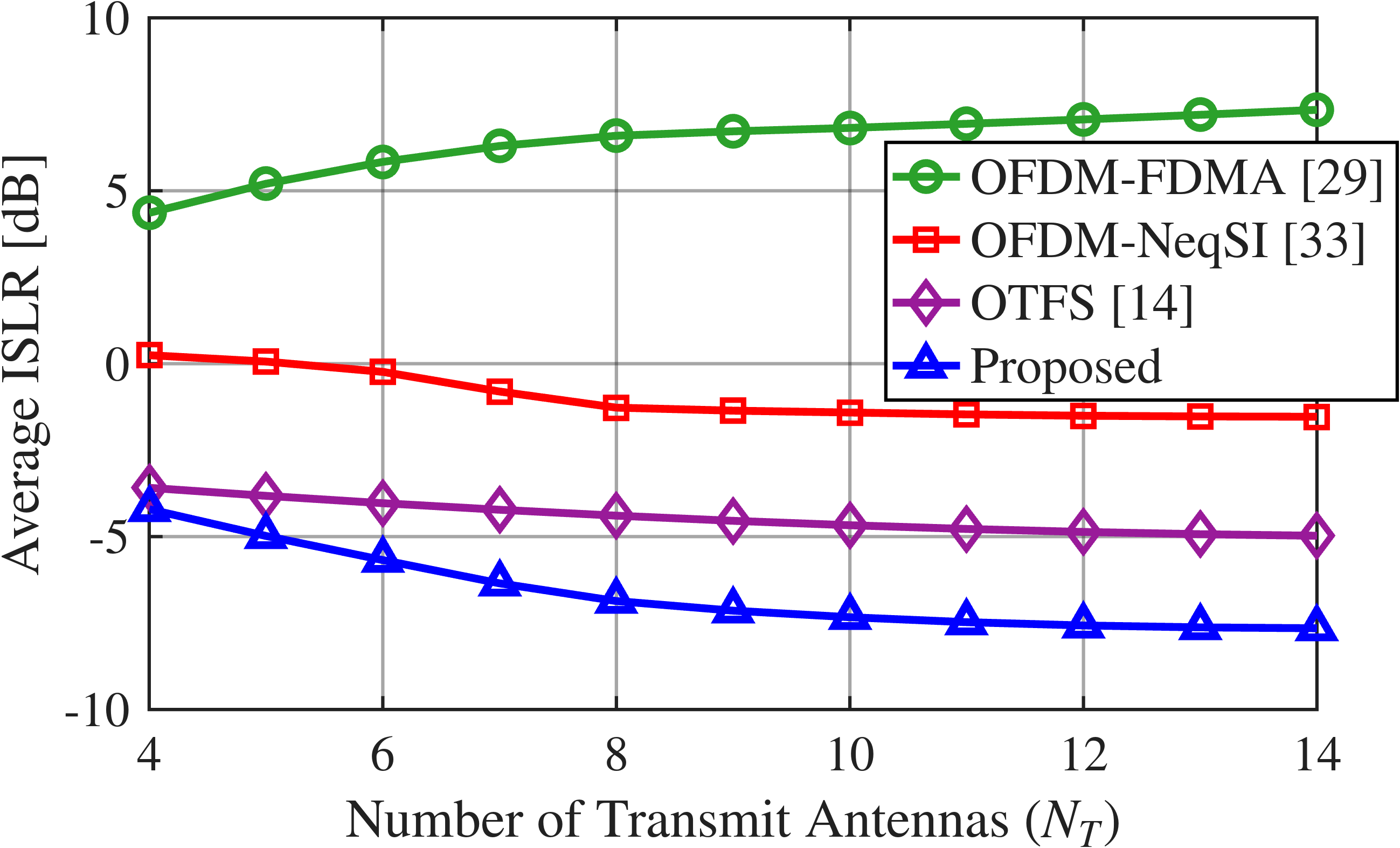}
\caption{\textcolor{black}{Average ISLR as a function of the number of transmit antennas $N_T$.}}
\label{fig:ISLR}
\end{figure}

We first evaluate the accuracy of the proposed pilot-aided multipath parameter estimation, since reliable data decoding directly depends on the quality of the reconstructed effective subspace channel. {The Normalized Mean Squared Error (NMSE) of the estimated effective channel, defined as $\mathbb{E}[\|\hat{\mathbf{G}}_{n,m} - \mathbf{G}_{n,m}\|_F^2] / \|\mathbf{G}_{n,m}\|_F^2$, is compared against the theoretical CRB derived in Sec.~\ref{subsec:param_est_perf}.}

Figure~\ref{fig:CRB} shows that the estimation error approaches the CRB as the SNR increases, confirming the statistical efficiency of the proposed estimator in the high-SNR regime. 
The overall trends for Profile~A and Profile~B are closely aligned, since in both cases the estimator operates on an equivalent linear observation model of the composite channel. 
The main difference lies in the overall channel energy, which results in a vertical shift of both the CRB and the corresponding NMSE curves, rather than in a change of convergence behavior.

Having established the reliability of the channel estimation stage, we next isolate the equalization capability by considering an ideal benchmark with perfect channel state information (CSI). In this case, the effective subspace channel is assumed to be known exactly at the receiver, as in Sec. IV-B, thereby eliminating estimation-induced errors and exposing only the impact of multipath-induced subspace coupling.
Figure~\ref{fig:mergedBer} reports the resulting bit error rate (BER) performance. The received signal $\mathbf{y}$ is processed using the linear combiners $\mathbf{A}_n$ designed according to the Zero-Forcing and Matched-Filtering criteria described in Sec.~IV. The perfect-CSI curve represents an upper bound, corresponding to decoding with exact knowledge of the orthogonal transmit subspaces. Although unattainable in practice, this benchmark quantifies the degradation caused by delay–Doppler distortions in multipath channels.
The proposed aggregate multipath-aware combiner significantly reduces inter-antenna interference and achieves BER performance close to the ideal reference across the entire SNR range, demonstrating the effectiveness of the structured linear equalization strategy.

\subsection{Imaging Performance}

We evaluate sensing and imaging performance using the framework described in Sec.~V. Figure~\ref{fig:radar_comparison} compares radar images obtained via pulse compression and TDBP under multipath conditions for {\color{black} four} waveform families in a multi-target scenario, which is reported in Fig. \ref{fig:profileB}.
The proposed Dual-Orthogonality waveform (Fig.~\ref{subfig:radar_DO}) accurately localizes all true targets, while multipath-induced replicas appear as structured ghost responses that remain spatially separable from the direct echoes.
In contrast, OFDM-FDMA (Fig.~\ref{subfig:radar_FDMA}) suffers from reduced range resolution due to bandwidth partitioning across transmit antennas. The resulting narrower per-transmitter bandwidth produces broader mainlobes and limits discrimination of closely spaced targets. OFDM-NeqSI (Fig.~\ref{subfig:radar_NeqSI}) restores full-band operation and achieves improved range resolution, confirming its role as a strong imaging baseline. However, its non-uniform spectral allocation increases sensitivity to frequency selectivity and residual synchronization errors under channel dynamics, leading to elevated sidelobes and structured artifacts in multi-target scenes.
{\color{black}
Similarly, despite operating over the full bandwidth, OTFS (Fig.~\ref{subfig:radar_OTFS}) suffers from severe ghosting. Multiplexing streams in the delay--Doppler domain induces high time-domain cross-correlation, manifesting as inter-stream leakage and false targets during TDBP processing.
}
The proposed Dual-Orthogonality waveform preserves full-band coherent imaging across all transmit channels without spectral fragmentation {\color{black} or delay-domain cross-ambiguity}. As a result, it achieves sharper focusing and improved separation between true scatterers and multipath-induced ghosts, particularly in the presence of closely spaced targets.

Figure~\ref{fig:SAR} further illustrates the focused MIMO-SAR reflectivity map obtained using Dual-Orthogonality waveforms. Coherent integration over 1024 aperture positions significantly enhances azimuth resolution and energy concentration. Multipath-induced ghost responses, such as second-bounce reflections, are suppressed by approximately 15~dB relative to the corresponding true scatterers. In realistic scenarios, higher-order multipath components typically exhibit lower energy due to additional propagation losses, further reinforcing the robustness of the proposed approach.

To quantitatively assess imaging quality and multipath artifact suppression, we adopt the Integrated Sidelobe Level Ratio (ISLR), defined as
\begin{equation}
\text{ISLR} = 10 \log_{10} \left( 
\frac{\int_{\mathcal{S}} |I(\mathbf{r})|^2 d\mathbf{r}}
{\int_{\mathcal{M}} |I(\mathbf{r})|^2 d\mathbf{r}}
\right),
\end{equation}
where $\mathcal{M}$ denotes the mainlobe region centered at the true target and $\mathcal{S}$ the remaining sidelobe region, such that a lower ISLR indicates superior suppression of multipath-induced ghost targets and clutter. Figure~\ref{fig:ISLR} reports the average ISLR as a function of the number of transmit antennas $N_T$. The proposed Dual-Orthogonality design consistently achieves the lowest ISLR and exhibits a monotonic improvement as $N_T$ increases. This behavior confirms that the additional spatial degrees of freedom are coherently exploited thanks to full-band orthogonality preservation. In contrast, the OFDM baselines exhibit performance saturation or degradation due to bandwidth fragmentation (FDMA) or residual non-orthogonality under Doppler (NeqSI).
{\color{black}
While OTFS outperforms both OFDM schemes via full-band transmission and improved delay--Doppler resolution, it hits an ISLR saturation floor due to intrinsic 2D grid sidelobes accumulating as inter-stream leakage. By mathematically forcing a zero-correlation zone, the proposed waveform circumvents this limit, achieving the lowest ISLR across all evaluated antenna configurations.
}

\subsection{Joint ISAC validation}

While ISLR characterizes sensing performance, communication efficiency must also be considered to evaluate overall ISAC capability. We therefore adopt the effective goodput metric
\begin{equation}
\mathcal{G} = \frac{(1 - B_e) K_e}{T_p} \quad [\text{Mbit/s}],
\end{equation}
where $B_e$ denotes the BER and $K_e$ is the number of information bits per waveform. This metric jointly captures reliability and spectral efficiency.
Figure~\ref{JointPlot} maps ISLR versus effective goodput as the number of transmit antennas varies, with the receive array fixed and SNR = 8~dB and 10~dB. Each point corresponds to a specific value of $N_T$ with $N_T=2, \ldots, 14$, and the associated goodput and ISLR are computed under identical channel conditions.
Three operating regions are defined with respect to theoretical performance limits: high-goodput values identify the communication-oriented region, low-ISLR values identify the sensing-oriented region, and the balanced ISAC region lies below the ISLR threshold and to the right of the goodput threshold.
The proposed Dual-Orthogonality waveform consistently achieves lower ISLR than {\color{black} both OFDM-NeqSI and OTFS} across all antenna configurations, indicating superior multipath suppression and improved imaging contrast. This sensing advantage is obtained with only a limited reduction in goodput, resulting in operating points that remain close to those of OFDM-NeqSI in terms of communication efficiency. OFDM-FDMA is omitted from the joint comparison because bandwidth partitioning across transmit antennas inherently reduces per-antenna resolution, leading to a fundamentally different sensing–communication operating regime that is not directly comparable with full-band schemes.

\begin{figure}[!t]
    \centering
    \subfloat[Profile A]{
        \includegraphics[width=0.9\linewidth]{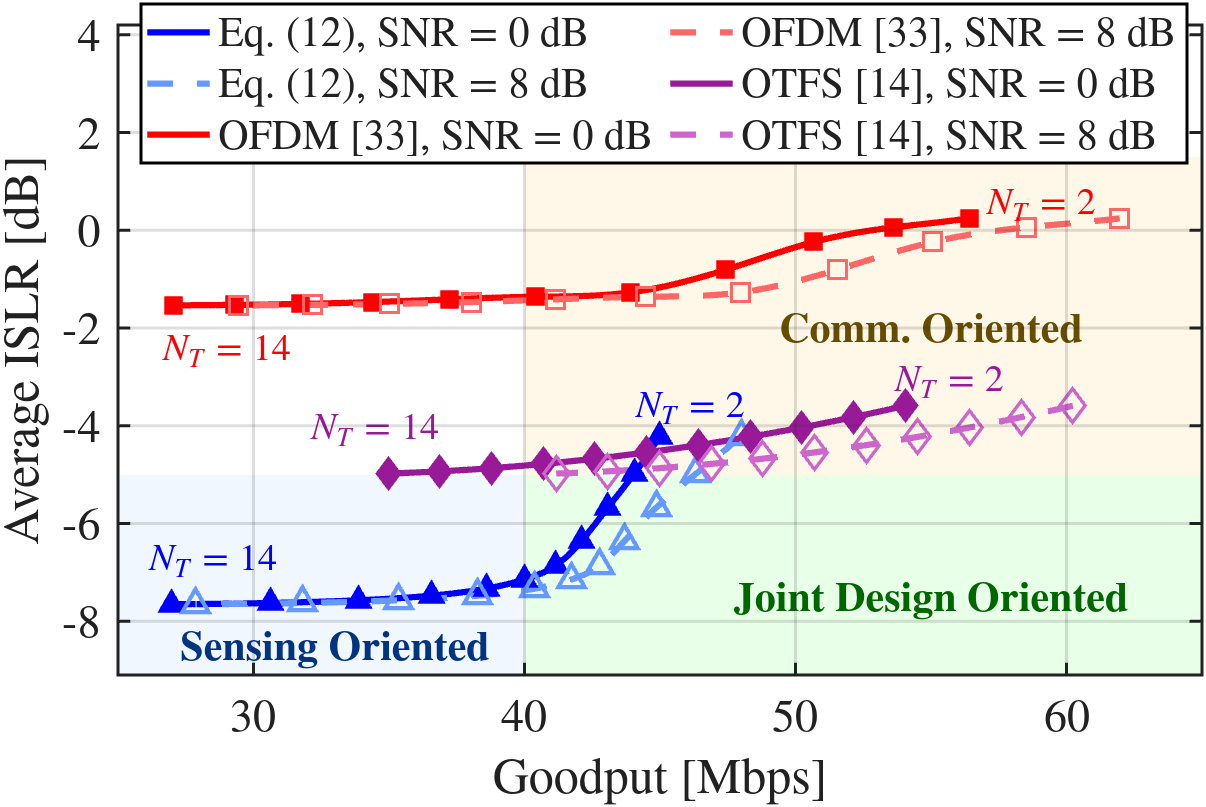}
        \label{subfig:joint_profileA}
    }
    
    \subfloat[Profile B]{
        \includegraphics[width=0.9\linewidth]{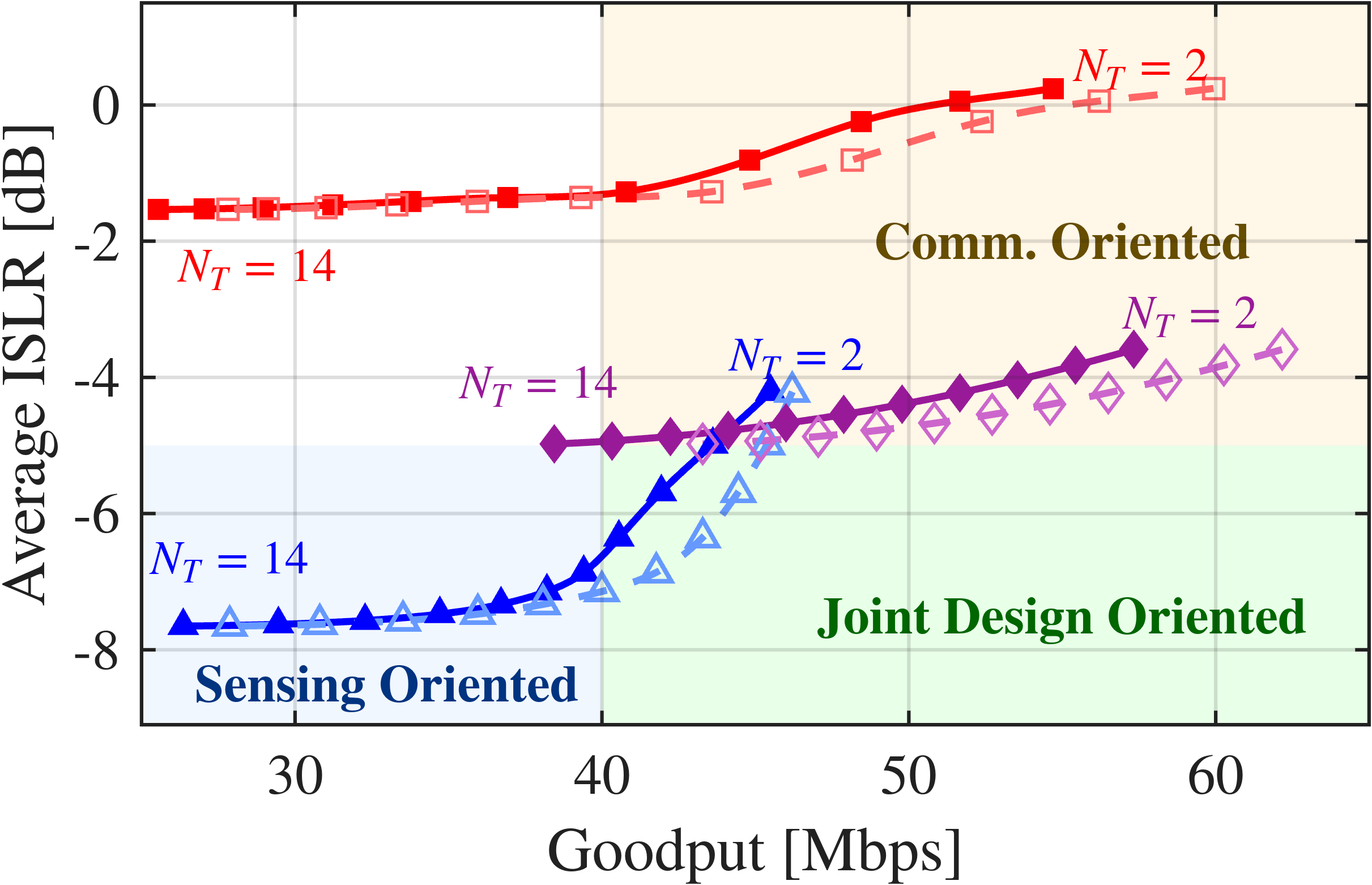}
        \label{xsubfig:joint_profileB}
    }
    {\color{black}
\caption{Joint sensing--communication performance, where the proposed waveform's communication goodput is evaluated on a single receive antenna, expressed in terms of ISLR and effective goodput for $N_T=2,\ldots,14$, with a fixed receive array and SNR values of 8~dB and 10~dB: (a) sparse multipath Profile~A and (b) rich multipath Profile~B. Each point corresponds to a different number of transmit antennas, comparing the proposed Dual-Orthogonality waveform with OFDM-NeqSI and OTFS under identical channel conditions.}
\label{JointPlot}
}
    
\end{figure}
Overall, the proposed scheme exhibits a more balanced ISAC behavior, reaching the joint design-oriented operating region as shown in the graphs, effectively balancing radar imaging and communication efficiency.
As the number of transmit antennas increases, the overall goodput decreases. For the Dual-Orthogonality waveform, this reduction is mainly due to the additional overhead introduced by the orthogonality constraints. In the case of NeqSI, the degradation is driven by enhanced sensitivity to Doppler and multipath effects, which become more pronounced with increasingly complex subcarrier allocation patterns required at higher antenna counts.
{\color{black}
Similarly, OTFS goodput degrades rapidly as $N_T$ grows due to severe inter-stream interference. Untangling these coupled streams overwhelms the receiver's linear equalizer, deteriorating the ISAC operating point. The proposed design inherently avoids this equalization bottleneck, maintaining a superior joint trade-off.
}

\section{Experimental Validation}

\begin{figure*}[!t]   
    
    \centering

\subfloat[Hardware setup\label{fig:hw_setup}]{%
        \includegraphics[width=0.2\textwidth]{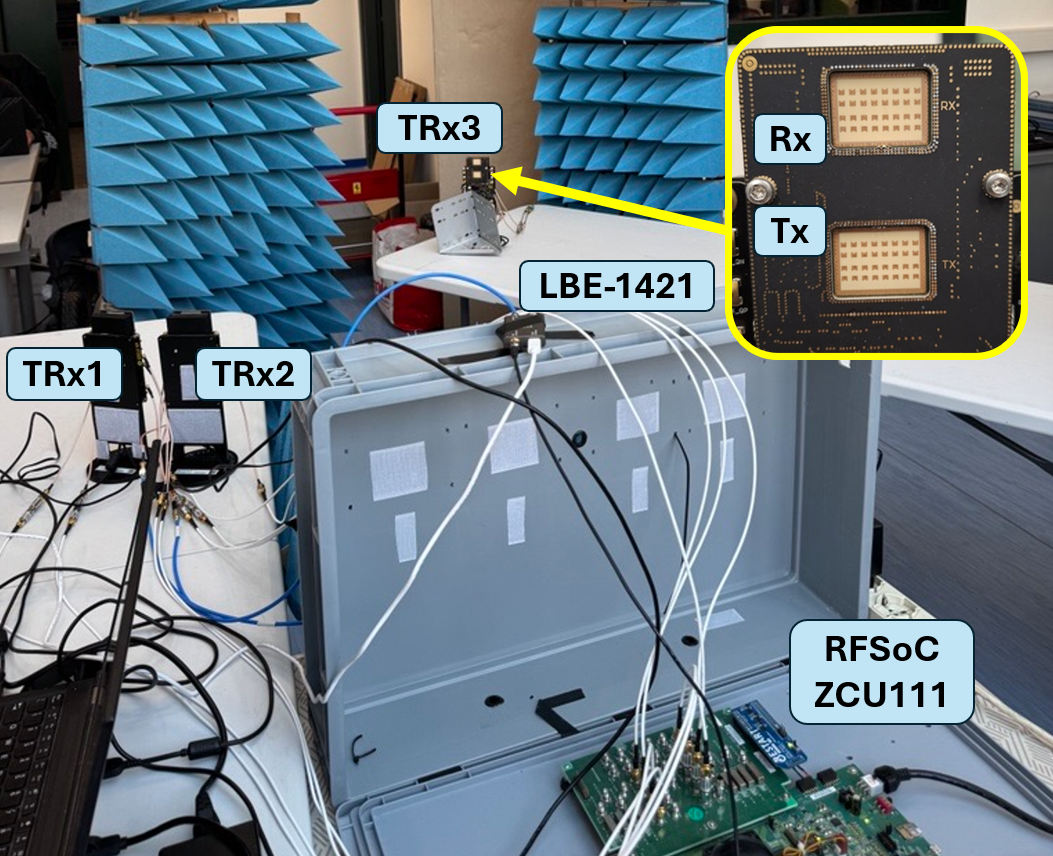}%
    }\hfill
    \subfloat[Tx1--Tx3 correlation\label{fig:tx1_tx3_corr}]{%
        \includegraphics[width=0.28\textwidth]{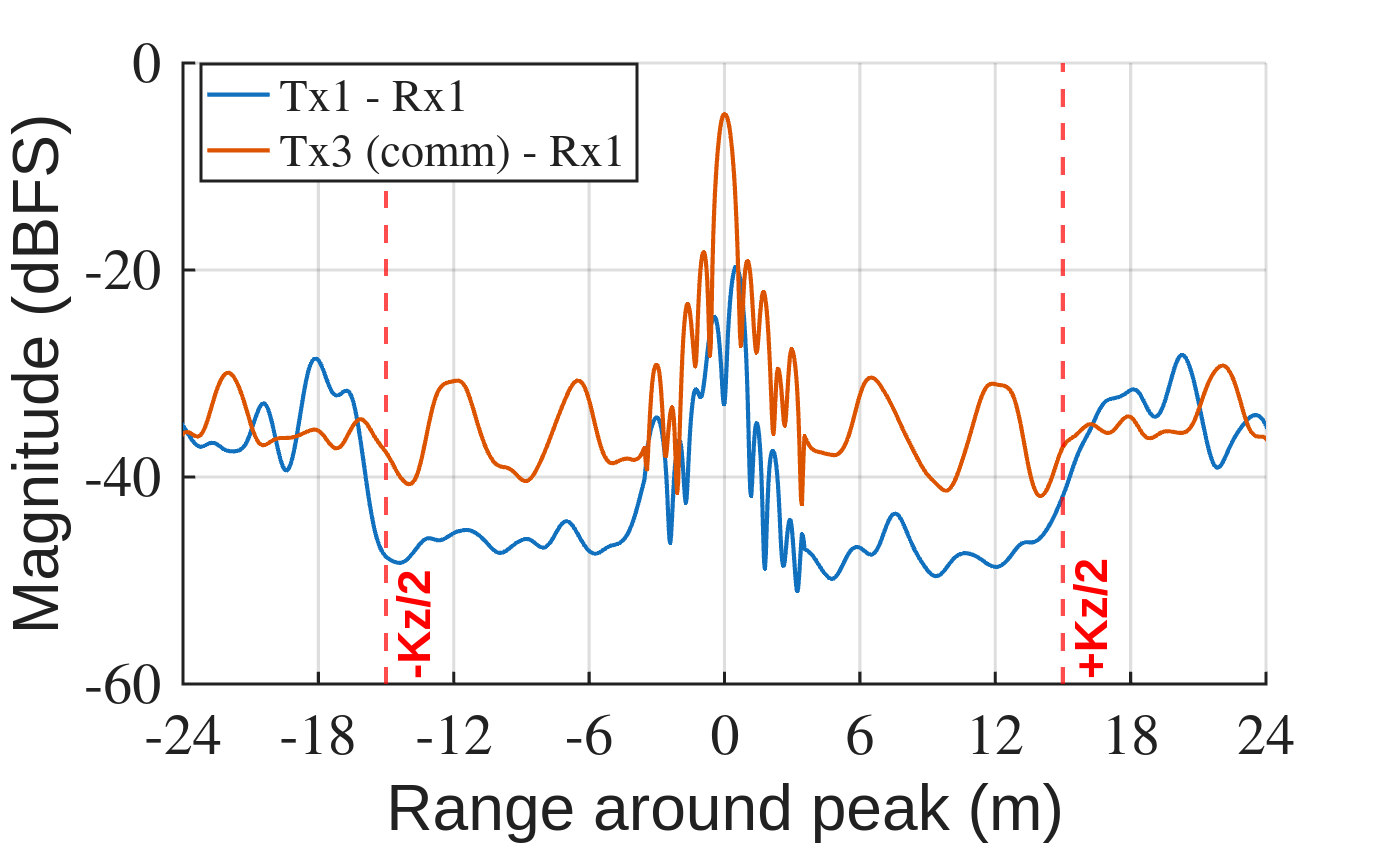}%
    }\hfill
    \subfloat[QPSK constellation\label{fig:qpsk_const}]{%
        \includegraphics[width=0.18\textwidth]{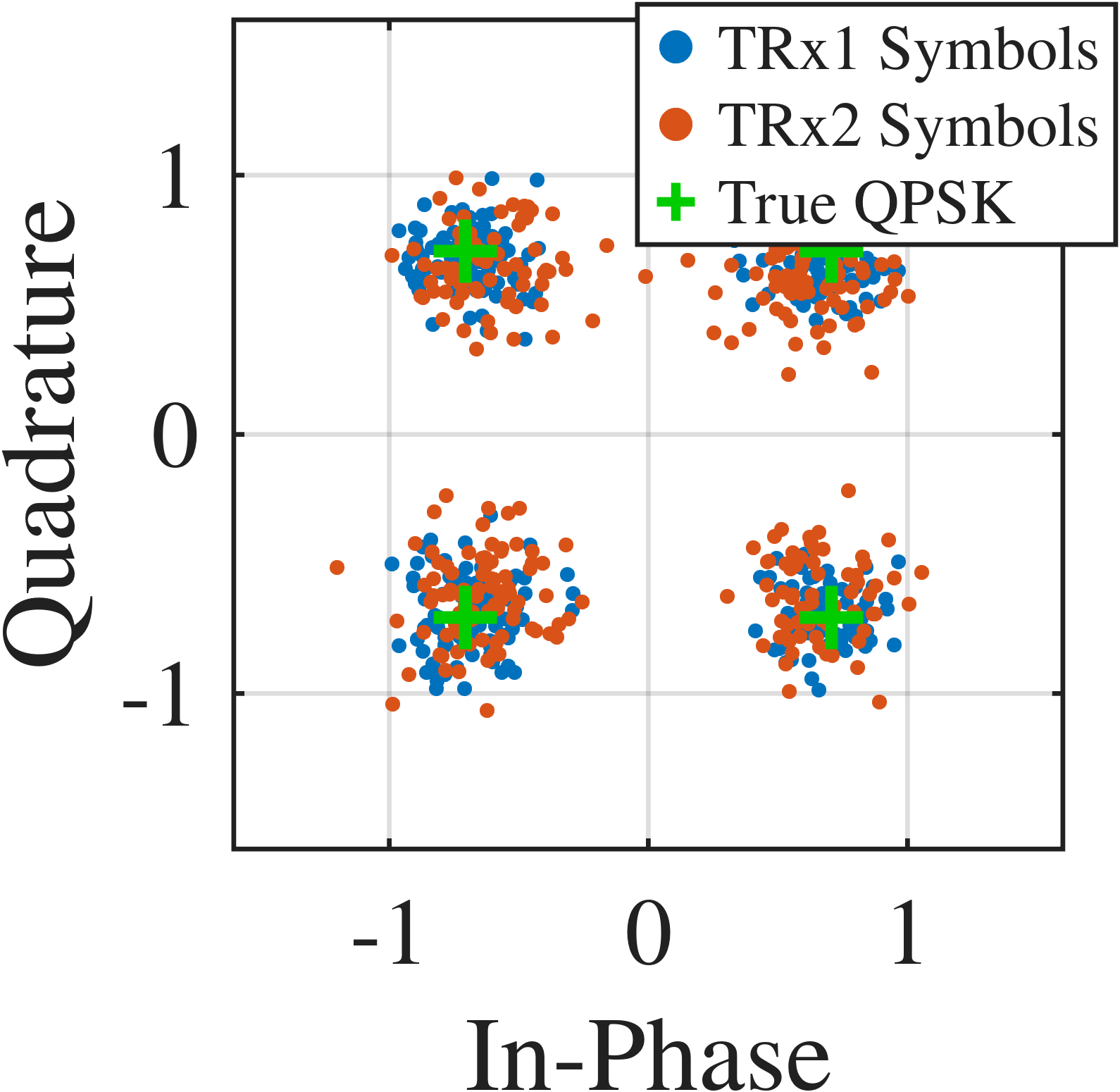}%
    }\hfill
    \subfloat[Sensing result\label{fig:sensing_result}]{%
        \includegraphics[width=0.29\textwidth]{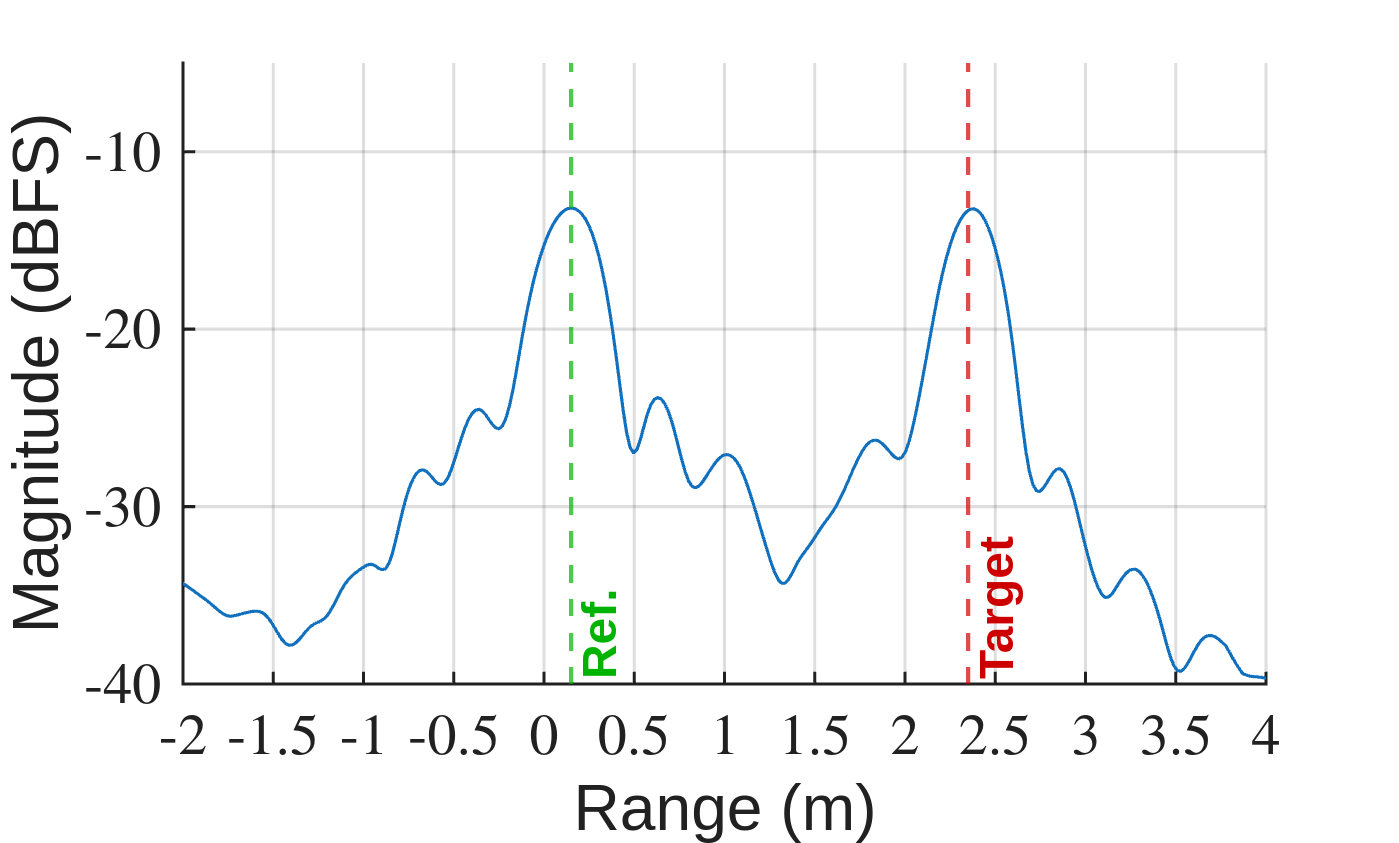}%
    }
    {\caption{\color{black}Experimental validation at 60~GHz:
(a) hardware setup comprising the RFSoC platform, three EVK06003 transceivers, synchronization hardware, anechoic panels, and corner reflectors; (b) Tx1--Tx3 cross-correlation showing the designed zero-correlation region; (c) recovered QPSK constellations for the two simultaneously transmitted streams; and (d) radar ranging result showing the reference and 2.35~m targets.}
\label{fig:experimental_results}
}
\end{figure*}

\textcolor{black}{This section experimentally validates the proposed ISAC waveforms across two configurations: Setup A (controlled Line-of-Sight) and Setup B (reflective multipath).
The experimental campaign focuses on the practical validation of waveform
properties, radar ranging, and multipath-aware communication processing,
while coherent MIMO-SAR imaging is assessed numerically in Section~VI.
Specifically, the measurements verify the designed zero cross-correlation
region, accurate radar ranging, and the effectiveness of multipath subspace
equalization under interference.
}

\subsection{Experimental Setup}

\textcolor{black}{The experimental setups are illustrated in Fig.~\ref{fig:experimental_results}(a) and Fig.~\ref{fig:setupB_validation}(a).} Measurements were conducted in the $60$~GHz Industrial, Scientific, and Medical (ISM) band using up/down-conversion modules with integrated antennas, specifically the EVK06003 evaluation kits~\cite{Sivers}. All modules were operated as transceivers (TRx) in full-duplex mode and driven by synchronized clocks in order to eliminate carrier frequency offsets. Clock synchronization was provided by an LBE-1421 GPS-disciplined oscillator.
While numerical results are obtained at $28$~GHz to align with standardized
cellular channel models, the experimental validation is conducted at $60$~GHz,
within the ISM band, due to regulatory availability and hardware constraints.
Since the proposed waveform design and receiver processing are formulated at baseband and depend on normalized delay--Doppler characteristics rather than on the carrier frequency itself, the experimental results are fully representative of the targeted operating regime and directly support the conclusions drawn from simulation.

Baseband waveform generation and acquisition were performed using an RFSoC ZCU111 platform~\cite{ZynqEK}, operated with the Multi-Tile Synchronization (MTS) reference design. The MTS framework enables precise time alignment between Analog-to-Digital Converters (ADCs) and Digital-to-Analog Converters (DACs), ensuring synchronization among transmitted waveforms. However, it constrains the system to real-valued waveforms and a maximum length of 16384 samples. The sampling frequency was set to $3932.16$~MHz to minimize the impact of up/down-sampling and ensure optimal system performance.

\textcolor{black}{The hardware configuration was adapted for each validation phase. In Setup A, three EVK06003 modules were employed: two co-located TRxs were used for sensing operations, while a third TRx was positioned facing them to evaluate baseline communication performance, with anechoic panels placed around the environment to suppress reflections. In Setup B, two transceivers (TRx1 and TRx2) were utilized to simultaneously transmit their respective data streams; the anechoic panels were removed to expose the reflective environment and intentionally induce severe multipath propagation.}
%

\begin{figure*}[!t]
\centering
\subfloat[Reflective experimental setup]{
    \includegraphics[width=0.39\textwidth]
    {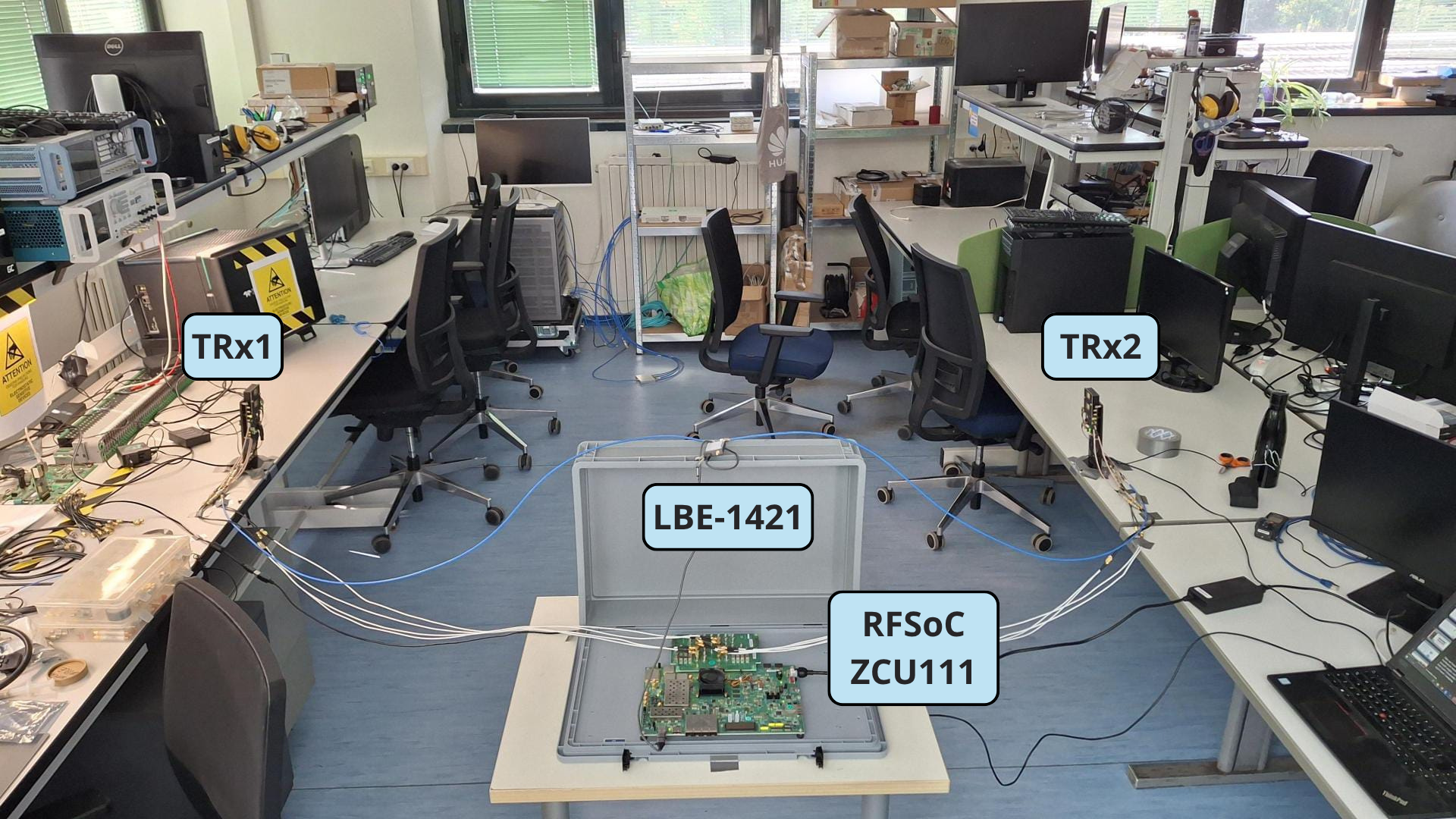}
}
\hfill
\subfloat[TRx1: Raw vs. equalized]{
    \includegraphics[width=0.24\textwidth]
    {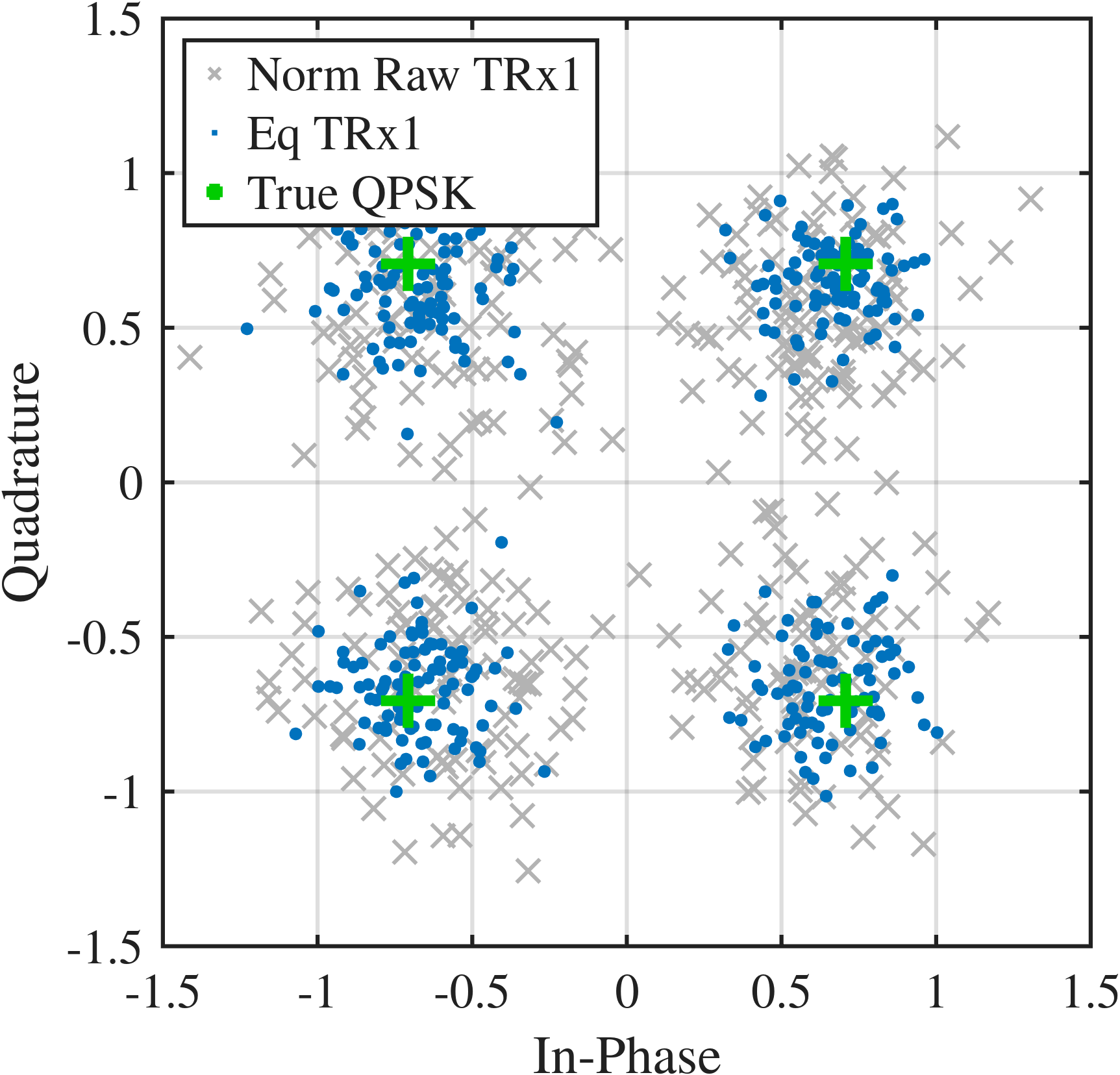}
}
\hfill
\subfloat[Decoded stream superposition]{
    \includegraphics[width=0.24\textwidth]
    {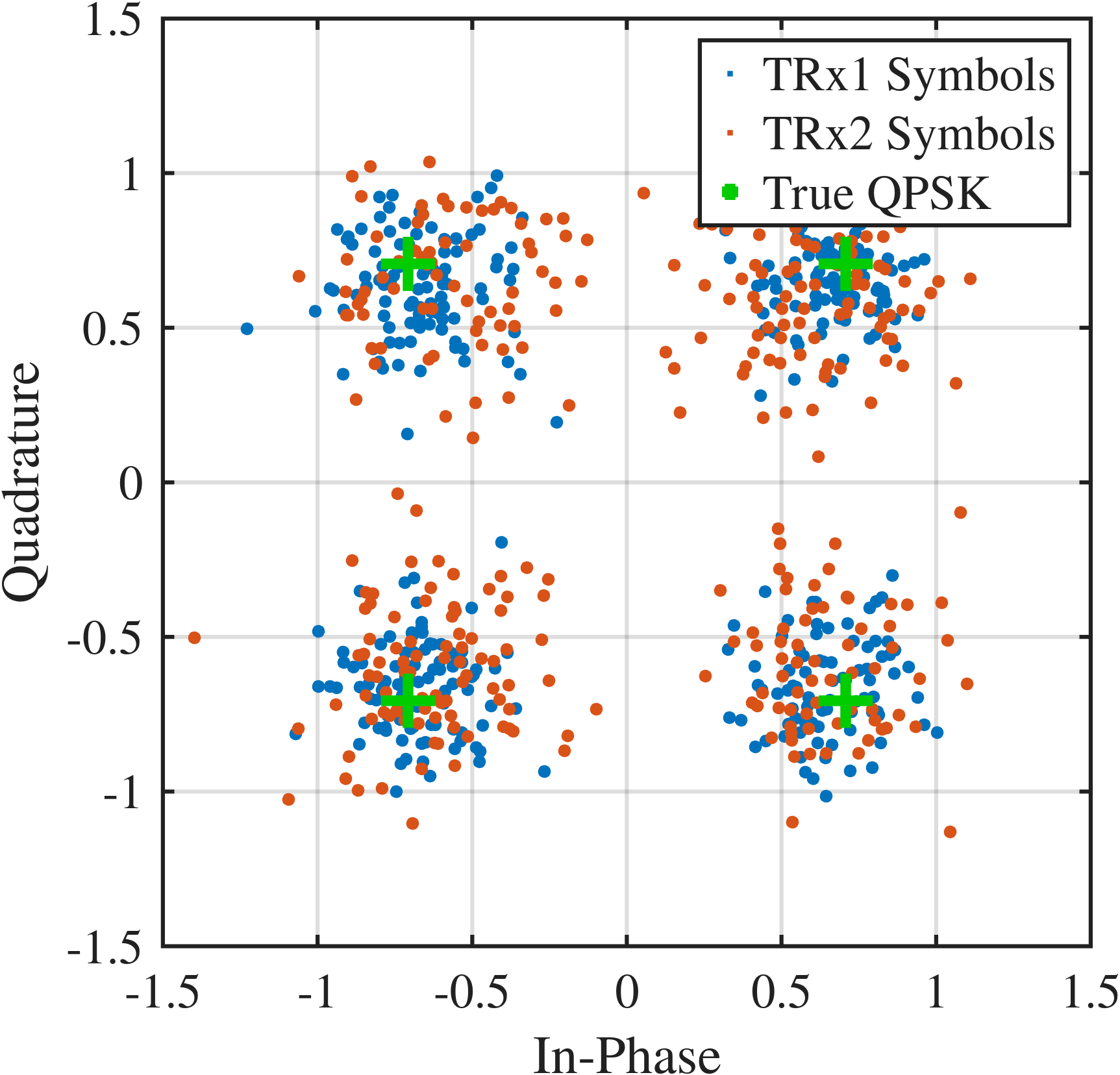}
}
\caption{\color{black}
Experimental validation in the reflective multipath environment
(Setup B) using a single receive antenna ($M_C=1$):
(a) laboratory configuration without anechoic absorbers;
(b) raw and equalized QPSK symbols for TRx1, showing the mitigation of
multipath-induced ISI and inter-stream interference; and
(c) superimposed decoded constellations, confirming the simultaneous
recovery of both transmitted streams from the same single-antenna
observation without receive-side spatial separation.
}
\label{fig:setupB_validation}
\end{figure*}

\textcolor{black}{\subsection{Baseline Configuration and Results}}

\textcolor{black}{In Setup A in Fig. \ref{fig:experimental_results}(a), the propagation environment allows isolating the intrinsic properties of the proposed waveform, such as inter-stream orthogonality, zero-autocorrelation regions, and stream separability, without confounding effects due to uncontrolled multipath (which are addressed in Setup B)}.

For the sensing experiments \textcolor{black}{in this configuration}, a corner reflector was positioned at a distance of $2.35$~m from the sensing TRxs to validate range estimation. Since acquisition was not synchronized with transmission start, a second corner reflector was placed close to the sensing TRxs to provide a zero-range reference at approximately $15$~cm.
For these communication experiments, a bandwidth of $245.76$~MHz was used, while sensing experiments employed a bandwidth of $491.52$~MHz. These values were selected as integer divisors of the sampling frequency. The resulting baseband waveforms consisted of $738$ and $1476$ samples, respectively. In both cases, the waveform duration was set to $3~\mu$s, corresponding to $11808$ samples at the DAC/ADC rate.

\textcolor{black}{The corresponding baseline results are} summarized in Fig.~\ref{fig:experimental_results}(b)--(d), covering both communication and sensing performance. From a communication perspective, Fig.~\ref{fig:experimental_results}(b) shows the measured cross-correlation between the transmitted waveforms of TRx1 and TRx3. The received signal at Rx1 consists of the superposition of the transmitted signals from all other TRxs filtered by their respective channels, as described in \eqref{comm_rec_equation}. By correlating the received signal with the known transmitted waveforms, a clear zero-autocorrelation region is observed. In this experiment, the zero-correlation interval was set to $\pm15$~m. Although hardware impairments and noise limit the depth of the correlation null, an average suppression of approximately $15$~dB is clearly visible, confirming the effectiveness of the proposed waveform design under practical conditions.
Communication performance was further evaluated at Rx3, which receives the superposition of the signals transmitted by TRx1 and TRx2, consistent with the signal model in \eqref{final_model} for $N_T=2$. The recovered Quadrature Phase Shift Keying (QPSK) constellations are shown in Fig.~\ref{fig:experimental_results}(c). The detected symbols are correctly clustered around the ideal constellation points, demonstrating effective stream separation. A slightly increased variance is observed for the second stream, which is attributed to the Dual-Orthogonality process and the associated noise enhancement during orthogonalization. Nevertheless, this effect does not compromise detection reliability. The measured Error Vector Magnitude (EVM) is $-12.89$~dB overall, with values of $-15.82$~dB and $-12.77$~dB for TRx1 and TRx2, respectively, confirming the high precision of the communication system.

Sensing performance is illustrated in Fig.~\ref{fig:experimental_results}(d). TRx1 and TRx2 were used to implement a monostatic radar configuration. The additional close-range corner reflector provides a reference peak at $15$~cm, compensating for the lack of timing synchronization between DAC generation and acquisition. The correlation outputs were interpolated by a factor of 16 to obtain smooth profiles and enable peak estimation with sub-resolution accuracy. For the employed sensing bandwidth of $491.52$~MHz, the theoretical range resolution is $30.5$~cm. The figure reports the average of four cross-correlation measurements obtained using both transmitted waveforms and receivers, clearly showing accurate detection of the target at $2.35$~m with centimeter-level precision.


\color{black}
\subsection{Validation in Reflective Multipath Environment}

\textcolor{black}{To complement the measurements of Setup A and directly evaluate the multipath-aware communication receiver, a second experimental campaign was conducted in a highly reflective indoor laboratory environment. The anechoic RF panels were removed, exposing the transmitted waveforms to strong unsuppressed reflections from metal shelving, monitors, glass surfaces, and hard desks, as shown in Fig.~\ref{fig:setupB_validation}(a). To quantify the severity of this reflective environment, we evaluated the Rician K-factor (defined as the ratio of the dominant path power to the scattered multipath power) from the extracted Channel Impulse Response (CIR). The controlled baseline (Setup A) exhibits a dominant Line-of-Sight (LOS) channel with a K-factor of $7.93$~dB. Conversely, Setup B introduces severe delay dispersion, driving the K-factor down to $2.79$~dB. Most notably, a prominent reflection is visible at a relative delay of $+5$ samples with an amplitude exceeding $50\%$ of the main LOS peak, confirming a highly dispersive propagation channel. In this setup, TRx1 and TRx2 simultaneously transmit their respective streams ($s_1$ and $s_2$). The receiver under test (Rx1, co-located at TRx1) captures the unshielded superposition of the incoming remote stream ($s_2$) and the multipath reflections of its own transmitted waveform ($s_1$).} Due to the distributed architecture of the transceiver modules, the
receiver Rx1 operates with a single-antenna observation ($M_C=1$).

Mathematically, for $M_C=1$, the receive-array term
$[\mathbf{a}_{\mathrm{rx}}(\phi_p)]_m$ reduces to a scalar and can be
absorbed into the corresponding effective path gain. As a consequence,
the spatial MUSIC and beamforming stages of the general multi-antenna
receiver are not available, and channel estimation is performed directly
from the temporal observation.
The estimated temporal channel parameters are then used to construct the
effective subspace channels required by the multipath-aware equalizer. Inter-stream interference is mitigated through projection onto the orthogonal complement of the estimated interference subspace,
$\boldsymbol{\Pi}^{\perp}_{\mathbf{H}^{\mathrm{int}}_{n,m}}$.
In this configuration, stream separation relies on the temporal and
code-domain structure of the Dual-Orthogonality waveforms rather than on receive-side spatial filtering.
The decoding results are reported in
Fig.~\ref{fig:setupB_validation}(b)--(c). In the reflective environment,
the raw symbols are strongly dispersed by ISI
superimposed on the IAI generated by the
simultaneously transmitted streams. The corresponding raw EVMs are $-8.31$~dB and $-7.85$~dB for TRx1 and TRx2,
respectively.

After applying the proposed multipath-aware equalizer, the EVM improves
to $-12.90$~dB for TRx1 and $-9.43$~dB for TRx2. As shown in
Fig.~\ref{fig:setupB_validation}(b), the equalized symbols concentrate
substantially closer to the ideal QPSK constellation points. Moreover,
Fig.~\ref{fig:setupB_validation}(c) shows that both transmitted streams
are simultaneously recovered from the same multipath-contaminated
single-antenna observation.

\color{black}

\section{Conclusion}

Dual-Orthogonality waveforms provide an effective solution for ISAC in dynamic multipath environments by enforcing orthogonality over a physically admissible delay window while preserving full-band operation per transmit channel. This avoids the bandwidth fragmentation of FDMA--OFDM and provides radar and imaging performance competitive with OFDM-NeqSI.
{\color{black}
Comparison with MIMO-OTFS further confirms favorable full-band sensing performance, with improved sidelobe behavior and a competitive joint sensing--communication operating point.
}
From a communication perspective, multipath-aware linear equalization effectively mitigates orthogonality loss, while the pilot-aided estimator approaches the Cramér--Rao Bound under favorable propagation conditions and remains robust in richer multipath scenarios. From a sensing perspective, coherent SAR processing provides accurate target localization and effective suppression of multipath-induced artifacts.
{\color{black}
While coherent MIMO-SAR imaging is assessed numerically, the $60$~GHz experimental campaign validates the waveform properties, radar ranging, and communication receiver. Measurements in a reflective environment further demonstrate multipath-aware recovery of simultaneously transmitted streams from a single-antenna observation without receive-side spatial separation, with clear EVM improvement after equalization.
}
Since the framework is formulated in terms of a generic unitary basis, OFDM-like signaling can be viewed as a particular instantiation. Hybrid designs combining OFDM communication efficiency with structured delay--Doppler subspace orthogonality represent a promising direction for future investigation.
\appendices

\section{Pilot Consistency under Dual-Orthogonality}
\label{app:pilot_consistency}

This appendix justifies the pilot insertion strategy adopted for multipath parameter estimation in Sec.~IV.

Let $\mathbf{C}_n\in\mathbb{C}^{K\times K_s}$ denote the code matrix associated
with the $n$-th transmit stream, and define 
$\mathbf{C}_{\mathrm{u}} \triangleq \mathbf{C}_n(1{:}K_u,:) \in \mathbb{C}^{K_u\times K_s},
 K_u < K_s$,
as the submatrix collecting its first $K_u$ rows. Assume that
$\mathrm{rank}(\mathbf{C}_{\mathrm{u}})=K_u$, and let
$\mathbf{C}_{\mathrm{u}}^{+}$ denote its Moore--Penrose pseudoinverse. Let
$\mathcal{N}_n\in\mathbb{C}^{K_s\times (K_s-K_u)}$ span the null space of
$\mathbf{C}_{\mathrm{u}}$, i.e., 
$\mathbf{C}_{\mathrm{u}} \mathcal{N}_n = \mathbf{0}$.
The symbol vector associated with the $n$-th transmit stream is constructed as
$\boldsymbol{\alpha}_n
=
\mathbf{C}_{\mathrm{u}}^{+}\mathbf{u}
+
\mathcal{N}_n \mathbf{x}_n$, 
where $\mathbf{u}\in\mathbb{C}^{K_u}$ denotes the known pilot sequence and
$\mathbf{x}_n\in\mathbb{C}^{K_s-K_u}$ collects the free information symbols.
The corresponding transmitted waveform
$\mathbf{s}_n=\mathbf{C}_n\boldsymbol{\alpha}_n$ satisfies
$[\mathbf{s}_n]_{1:K_u}
=
\mathbf{C}_{\mathrm{u}} \boldsymbol{\alpha}_n
=
\mathbf{u}$,
since $\mathbf{C}_{\mathrm{u}} \mathbf{C}_{\mathrm{u}}^{+} = \mathbf{I}_{K_u}$ and
$\mathbf{C}_{\mathrm{u}} \mathcal{N}_n = \mathbf{0}$.

{\color{black}
While pilots can theoretically be embedded in any stream $n$ provided the intersection $\mathrm{Null}(\mathbf{C}_{\mathrm{u}}) \cap \mathrm{Null}(\mathbf{Q}_n)$ is non-empty, our implementation assigns the pilot exclusively to the first stream ($n=1$, where $\mathbf{Q}_1=\mathbf{0}$) to minimize overhead and avoid subspace conflicts. Consequently, the available payload dimension for the first stream reduces to $K_s - K_u$, whereas all subsequent streams ($n>1$) carry no pilots and retain their full data dimension $d_n = K_s - (n-1)(K_z - 1)$. In our numerical evaluations, a pilot of length $K_u = 30$ is used. For $B=40$~MHz ($K_s=200$), this introduces a $15\%$ overhead on the first stream and exactly $0\%$ on all remaining streams.
}
\color{black}
\section{Full-Rank Condition for Dominant-Path Interference Nulling}
\label{app:rank_proof}

This appendix derives the rank condition underlying the dominant-path interference-nulling construction in Sec.~IV, formalizing when linear subspace projection can suppress inter-stream interference from a dominant multipath component.


Consider the path-dependent delay--Doppler operator $\mathbf{H}_{\hat p}$
associated with the dominant multipath component indexed by $\hat p$, and define
the stacked interference matrix corresponding to the $n$-th transmit stream as
$ \mathbf{H}^{\mathrm{int}}_{\hat{p},n}=
\big[
\mathbf{H}_{\hat{p}}\mathbf{C}_{1},
\dots,
\mathbf{H}_{\hat{p}}\mathbf{C}_{n-1},
\mathbf{H}_{\hat{p}}\mathbf{C}_{n+1},
\dots,
\mathbf{H}_{\hat{p}}\mathbf{C}_{N_T}
\big]
\in
\mathbb{C}^{K \times (N_T-1)K_s}
$,
where $\mathbf{C}_\ell\in\mathbb{C}^{K\times K_s}$ denotes the orthogonal code
matrix assigned to the $\ell$-th transmit stream, and the block corresponding to
$\ell = n$ is omitted.

The delay--Doppler operator $\mathbf{H}_{\hat p}$ admits the decomposition
$\mathbf{H}_{\hat p}
=
\mathbf{D}(\nu_{\hat p})\,
\mathbf{F}^{\mathsf{H}}\mathbf{B}(\tau_{\hat p})\,
\mathbf{F}
\;\triangleq\; \mathbf{U}$, 
where $\mathbf{F}$ is the unitary DFT matrix and $\mathbf{B}(\cdot)$ and
$\mathbf{D}(\cdot)$ are diagonal phase-rotation matrices modeling delay and
Doppler, respectively. Since $\mathbf{U}$ is the product of unitary matrices, it
is itself unitary and therefore rank preserving.
Consequently, the rank of the interference matrix satisfies
$\mathrm{rank}\!\left(\mathbf{H}^{\mathrm{int}}_{\hat{p},n}\right)
=
\mathrm{rank}\!(
\big[\,
\mathbf{C}_\ell
\,\big]_{\ell \neq n}
)$.
From the Dual-Orthogonality construction in Sec.~III, the matrices
$\{\mathbf{C}_\ell\}$ have orthonormal columns and mutually orthogonal column
spaces. It follows that
$\mathrm{rank}\!(
\big[\,
\mathbf{C}_\ell
\,\big]_{\ell \neq n}
)
=
\sum_{\ell \neq n} \mathrm{rank}(\mathbf{C}_\ell)
=
(N_T-1)K_s$.
The interference matrix $\mathbf{H}_{n,\hat p}^{\mathrm{int}}$ has full column rank and admits a non-trivial left null space provided that $(N_T-1)K_s < K$.
This condition is satisfied under the design choice $K_s \approx K/N_T$ adopted
throughout the paper. While the above inequality provides a sufficient (but not
necessary) condition, it guarantees the existence of a non-trivial null space.

{\color{black}
\section{Validity of the Circulant Channel Approximation}
\label{app:circulant_bound}

This appendix derives the mathematical bound on the error introduced by approximating the true physical linear (Toeplitz) multipath delay with the circulant operator used in \eqref{final_model}.

Let $\mathbf{S} \in \mathbb{C}^{K \times N_T}$ denote the transmit block for the current signaling interval, and $\mathbf{S}_{\text{prev}}$ denote the block from the previous interval, both containing independent zero-mean symbols with variance $\sigma_s^2$. For a generic path $p$ with delay $\tau_p$, let $L_p = \lceil \tau_p / T_s \rceil$ denote the discrete sample delay. 

A true physical delay applies a linear shift. Let $\mathbf{T}_{L_p} \in \mathbb{C}^{K \times K}$ denote the linear zero-padded integer-delay operator (as defined in Proposition 1), which shifts the block down by $L_p$ samples. Let $\mathbf{T}_{\text{wrap}, L_p}$ denote the complementary upper-shift matrix that extracts the last $L_p$ samples. The exact physical received signal for the current block, including Inter-Block Interference (IBI) from the previous block, is
\begin{equation}
    \mathbf{Y}_{\text{exact}} = \sum_{p=0}^{P-1} \gamma_p \mathbf{D}(\nu_p) \Big[ \mathbf{T}_{L_p} \mathbf{S} + \mathbf{T}_{\text{wrap}, L_p} \mathbf{S}_{\text{prev}} \Big] \mathbf{A}_p + \mathbf{W},
\end{equation}
where $\mathbf{A}_p = \mathbf{a}_{\text{tx}}(\theta_p)\mathbf{a}_{\text{rx}}^{\mathsf{T}}(\phi_p)$. 

Conversely, the circulant approximation $\mathbf{F}^{\mathsf{H}}\mathbf{B}(\tau_p)\mathbf{F}$ mathematically applies a circular shift, which wraps the tail of the current block to the front. The circulant operator decomposes perfectly as $\mathbf{T}_{L_p} + \mathbf{T}_{\text{wrap}, L_p}$. The approximated model in \eqref{final_model} is thus equivalent to
\begin{equation}
    \mathbf{Y}_{\text{circ}} = \sum_{p=0}^{P-1} \gamma_p \mathbf{D}(\nu_p) \Big[ \mathbf{T}_{L_p} \mathbf{S} + \mathbf{T}_{\text{wrap}, L_p} \mathbf{S} \Big] \mathbf{A}_p + \mathbf{W}.
\end{equation}

Defining the approximation error matrix as $\mathbf{E} = \mathbf{Y}_{\text{circ}} - \mathbf{Y}_{\text{exact}}$, the main linear delay terms and noise cancel exactly, yielding $\mathbf{E} = \sum_{p=0}^{P-1} \mathbf{E}_p$, where the error induced by path $p$ is explicitly defined as:
\begin{equation}
    \mathbf{E}_p = \gamma_p \mathbf{D}(\nu_p) \mathbf{T}_{\text{wrap}, L_p} \big( \mathbf{S} - \mathbf{S}_{\text{prev}} \big) \mathbf{A}_p.
\end{equation}
This confirms the error is confined solely to the $L_p$ boundary samples, representing the mismatch between the assumed circular wrap-around and the true IBI.

To bound the total expected error energy, we apply the standard Uncorrelated Scattering assumption, where independent paths possess uncorrelated phases ($\mathbb{E}[\gamma_p \gamma_{p'}^*] = 0$ for $p \neq p'$). Consequently, cross-path terms vanish in expectation, and the total error energy is the sum of the individual path energies. For a given path $p$, since $\mathbf{D}(\nu_p)$ is unitary, $\|\mathbf{A}_p\|_F^2 = N_T M_C$, and $\mathbf{T}_{\text{wrap}, L_p}$ extracts exactly $L_p$ rows of the independent data difference (variance $2\sigma_s^2$), we have
\begin{equation}
    \mathbb{E}[\|\mathbf{E}_p\|_F^2] = |\gamma_p|^2 N_T M_C (2 L_p \sigma_s^2).
\end{equation}
The expected signal energy received from path $p$ without error is $\mathbb{E}[\|\mathbf{Y}_{\text{exact}, p}\|_F^2] = |\gamma_p|^2 N_T M_C K \sigma_s^2$. 

The relative error energy is given by the ratio:
\begin{equation}
    \frac{\mathbb{E}[\|\mathbf{E}_p\|_F^2]}{\mathbb{E}[\|\mathbf{Y}_{\text{exact}, p}\|_F^2]} = \frac{2 L_p}{K} \le \frac{2(\tau_{\max} + T_s)}{T_p},
\end{equation}
where $\tau_{\max}$ is the maximum channel delay spread. Note that this exact decomposition assumes integer sample delays; for fractional delays, the operator $\mathbf{F}^{\mathsf{H}}\mathbf{B}(\tau_p)\mathbf{F}$ does not reduce to a perfect integer circular shift, which constitutes an inherent limitation of the compact circulant model. Nevertheless, for the system parameters considered in this work ($T_p = 40\,\mu\text{s}$, $\tau_{\max} \approx 1\,\mu\text{s}$), the maximum relative error energy introduced by the circulant approximation is bounded by $\approx 5\%$. Therefore, because $T_p \gg \tau_{\max}$, the circulant model accurately captures the physical channel macroscopic behavior while enabling compact frequency-domain equalizer design.

}

\color{black}

\bibliographystyle{IEEEtran}
\bibliography{bibliography}

\end{document}